\documentclass[hidelinks,onefignum,onetabnum]{siamart251104}

\usepackage{lipsum}
\usepackage{amsfonts}
\usepackage{graphicx}
\usepackage{epstopdf}
\usepackage{algorithmic}
\usepackage{pifont}
\usepackage{enumitem}
\usepackage{amssymb}
\usepackage{pdflscape}
\usepackage{braket}

\ifpdf
  \DeclareGraphicsExtensions{.eps,.pdf,.png,.jpg}
\else
  \DeclareGraphicsExtensions{.eps}
\fi

\newsiamremark{remark}{Remark}
\newsiamremark{hypothesis}{Hypothesis}
\crefname{hypothesis}{Hypothesis}{Hypotheses}
\newsiamthm{claim}{Claim}
\newsiamremark{fact}{Fact}
\crefname{fact}{Fact}{Facts}

\newsiamthm{problem}{Problem}

\newcommand{\C}{\mathbb{C}}
\newcommand{\R}{\mathbb{R}}

\newcommand{\Tr}{\operatorname{Tr}}

\newcommand{\Skew}{\operatorname{Skew}}
\newcommand{\Exp}{\operatorname{Exp}}

\newcommand{\grad}{\operatorname{grad}}

\headers{Grover-Compatible Optimization for Quantum Search}{Z. Lai, D. An, J. Hu, and Z. Wen}

\title{
A Grover-compatible Manifold Optimization Algorithm for Quantum Search\thanks{Submitted to the editors DATE. $^\triangle$Corresponding author. \funding{This work was supported by the National Natural Science Foundation of China under the grant numbers 12501419, 12331010 and 12288101, the National Key R\&D Program of China under the grant number 2024YFA1012901, the Quantum Science and Technology-National Science and Technology Major Project via Project 2024ZD0301900, and the Fundamental Research Funds for the Central Universities, Peking University.}
}
}

\author{
  Zhijian Lai\thanks{Beijing International Center for Mathematical Research, Peking University, Beijing, 100871, People's Republic of China (\email{lai\_zhijian@pku.edu.cn,
    dongan@pku.edu.cn,
    wenzw@pku.edu.cn}).
}
  \and
  Dong An\footnotemark[2]
  \and
  Jiang Hu\thanks{Yau Mathematical Sciences Center,
    Tsinghua University, Beijing, 100190, People's Republic of China (\email{jianghu@tsinghua.edu.cn}).
  }
  \and
  Zaiwen Wen\footnotemark[2] $^\triangle$
}

\usepackage{amsopn}

\usepackage{xcolor}
\definecolor{mygreen}{RGB}{205, 222, 194}

\usepackage{array}
\newcolumntype{C}[1]{>{\centering\arraybackslash}m{#1}}
\newcolumntype{L}[1]{>{\raggedright\arraybackslash}m{#1}}

\ifpdf
\hypersetup{
  pdftitle={A Grover-compatible Manifold Optimization Algorithm for Quantum Search},
  pdfauthor={Zhijian Lai, Dong An, Jiang Hu, and Zaiwen Wen}
}
\fi

\begin{document}

\maketitle

\begin{abstract}
Grover's algorithm is a fundamental quantum algorithm that offers a quadratic speedup for the unstructured search problem by alternately applying physically implementable oracle and diffusion operators.
In this paper, we reformulate the unstructured search as a maximization problem on the unitary manifold and solve it via the Riemannian gradient ascent (RGA) method.
To overcome the difficulty that generic RGA updates do not, in general, correspond to physically implementable quantum operators, we introduce Grover-compatible retractions to restrict RGA updates to valid oracle and diffusion operators.
Theoretically, we establish a local Riemannian $\mu$-Polyak--\L{}ojasiewicz (PL) inequality with $\mu = \tfrac{1}{2}$, which yields a linear convergence rate of $1 - \kappa^{-1}$ toward the global solution. Here, the condition number $\kappa = L_{\mathrm{Rie}} / \mu$, where $L_{\mathrm{Rie}}$ denotes the Riemannian Lipschitz constant of the gradient. Taking into account both the geometry of the unitary manifold and the special structure of the cost function, we show that $L_{\mathrm{Rie}} = \mathcal{O}(\sqrt{N})$ for problem size $N = 2^n$. Consequently, the resulting iteration complexity is $\mathcal{O}(\sqrt{N} \log(1/\varepsilon))$ for attaining an $\varepsilon$-accurate solution, which matches the quadratic speedup of $\mathcal{O}(\sqrt{N})$ achieved by Grover's algorithm.
These results demonstrate that an optimization-based viewpoint can offer fresh conceptual insights and lead to new advances in the design of quantum algorithms.
\end{abstract}



\section{Introduction}

Unstructured search represents a fundamental challenge in computational science, finding broad applications in data retrieval \cite{knuth1997art,manning2008introduction}, optimization \cite{durr1996quantum,baritompa2005grover,gilliam2021grover}, cryptanalysis \cite{grassl2016applying,bernstein2025post}, and machine learning \cite{dong2008quantum,du2021grover}. Its goal is to identify one or more target elements within an unsorted search space of size $N$. In the absence of exploitable structure, any classical algorithm inevitably requires $\Omega(N)$ queries in the worst case \cite{bennett1997strengths}, constituting a critical bottleneck for large-scale search.

In 1996, L. K. Grover introduced a groundbreaking quantum algorithm, known as Grover's algorithm \cite{grover1996fast}, which solves the unstructured search problem with a query complexity of $\mathcal{O}(\sqrt{N})$. This quadratic speedup serves as a landmark demonstration of the power of quantum computation over classical algorithms and has become a paradigmatic example of quantum advantage.
Moreover, this $\mathcal{O}(\sqrt{N})$ query complexity has been rigorously shown to be optimal \cite{nielsen2010quantum,zalka1999grover,bennett1997strengths,beals2001quantum},
establishing a tight bound of $\Theta(\sqrt{N})$.
The significance of Grover's algorithm extends far beyond the specific problem of unstructured search.
It reveals a fundamental \textit{primitive} of quantum computation.
Its underlying mechanism was later generalized to the amplitude amplification (AA) technique \cite{brassard2000quantum,ambainis2012variable,suzuki2020amplitude}, an indispensable tool for boosting the success probability of various quantum algorithms \cite{durr1996quantum,szegedy2004quantum,biamonte2017quantum,acampora2022using}.
These developments were subsequently unified within the broader framework of quantum singular value transformation (QSVT) \cite{gilyen2019quantum,martyn2021grand}.

To understand the origin of Grover's quadratic speedup, several geometric and dynamical interpretations have been proposed. For instance, Grover's iteration has been characterized as a geodesic in complex projective space connecting the initial and target states \cite{miyake2001geometric}, and modeled as geodesic dynamics within an information-geometric framework under the Wigner--Yanase metric \cite{cafaro2012grover}.
More recently, an interpretation based on imaginary-time evolution (ITE) has been introduced in \cite{suzuki2025grover}.

Manifold (Riemannian) optimization \cite{absil2008optimization,boumal2023introduction,hu2020brief} emerges as a compelling framework for understanding Grover's algorithms. Indeed, manifolds arise naturally in quantum systems, and manifold-based methods have already been successfully applied to tasks such as optimizing quantum circuits \cite{luchnikov2021qgopt,wiersema2023optimizing} and quantum state and comb tomography \cite{hsu2024quantum,li2025quantum}. In particular, the set of quantum circuits constitutes the unitary manifold $\{U \in \mathbb{C}^{N \times N} \mid U^{\dagger} U = I\}$.
Since many quantum tasks aim to find a circuit that prepares a target state or implements a desired transformation, they can be naturally formulated as optimization problems over this unitary manifold \cite{wiersema2023optimizing}.
A series of double-bracket quantum algorithms \cite{gluza2024double,zander2025role,suzuki2025double,gluza2024doublediag,robbiati2024double,xiaoyue2024strategies} shows that the ITE trajectories follow a Riemannian gradient flow on the unitary manifold, induced by an associated manifold optimization problem.
These works suggest that manifold optimization can potentially provide a promising and systematic tool for analyzing quantum algorithms.

Recent developments in manifold optimization indicate that algorithm design and convergence theory depend critically on the choice of a \textit{retraction} \cite{adler2002newton,absil2008optimization}, the operator that maps a tangent vector back onto the manifold and thereby yields the fundamental update step. On the unitary manifold, a natural retraction is the Riemannian exponential\footnote{The Riemannian exponential map is a geometric concept extending straight lines to manifolds (i.e., geodesics).} map \cite{absil2008optimization,abrudan2008efficient}, which leads to the update step $U_{k+1} = e^{i H t_k} U_k$.
In the context of classical computing, this choice is often replaced by computationally more efficient retractions, such as the Cayley transform \cite{wen2013feasible}, QR-based retractions \cite{sato2019cholesky,absil2008optimization}, or polar decompositions \cite{absil2012projection,absil2008optimization}, to avoid explicit matrix exponentials.
However, these alternative retractions typically rely on matrix inversions or factorizations that are difficult to realize on quantum hardware, whereas the exponential factor $e^{i H t_k}$ is exactly a quantum gate.
Consequently, many retractions that are attractive for a classical computer may not be physically implementable on a quantum computer.
Despite the intrinsic suitability of manifold optimization for quantum computing, existing results neither provide an interpretation of Grover-type iterations nor yield Riemannian algorithms with rigorously proven quadratic speedups.
This motivates us to develop manifold optimization algorithms whose retractions are compatible with Grover's quantum gates and to establish convergence guarantees that recover Grover's quadratic speedup.

\subsection{Our contribution}
Building on the manifold optimization viewpoint outlined above, our contributions are twofold:
\begin{itemize}
\item \textbf{Grover-compatible Riemannian gradient ascent algorithms.} We first reformulate the unstructured search problem as a maximization problem on the unitary manifold. We then introduce several \emph{Grover-compatible retractions} on this manifold and employ them to design Riemannian gradient ascent (RGA) algorithms. By construction, each update of these RGA methods is implemented as a finite composition of oracle and diffusion operators. This design ensures that the updates are physically realizable on quantum hardware while simultaneously constituting valid Riemannian retractions within the manifold optimization framework.

\item \textbf{Linear convergence with Grover's quadratic speedup.}
To analyze the complexity of our algorithm, we first establish a local Riemannian $\mu$-Polyak--\L{}ojasiewicz (PL) inequality with parameter $\mu = \tfrac{1}{2}$ for the proposed cost function.
This property ensures a linear convergence rate towards the global solution with a per-iteration contraction factor of $1 - \kappa^{-1}$, where the condition number is given by $\kappa = L_{\mathrm{Rie}} / \mu$ and $L_{\mathrm{Rie}}$ denotes the Riemannian Lipschitz constant of the gradient.
By carefully exploiting the geometry of the unitary manifold and the structure of the cost function, we prove that $L_{\mathrm{Rie}} = \mathcal{O}(\sqrt{N})$ for problem size $N$. Consequently, with a step size of order $1 / L_{\mathrm{Rie}}$, the resulting Grover-compatible RGA algorithms attain an $\varepsilon$-accurate solution in $\mathcal{O}(\sqrt{N} \log(1/\varepsilon))$ iterations.
This result matches Grover's quadratic speedup in terms of $N$, while achieving linear convergence with the optimal error dependence $\log(1/\varepsilon)$ under the PL condition.

\end{itemize}

\subsection{Organization}
In \cref{sec-prelim}, we review the unstructured search problem and the basic geometry of the unitary manifold.
In \cref{sec-grover-alg}, we reformulate the search problem as an optimization problem, introduce the Grover-compatible retraction, and derive the associated Riemannian gradient ascent method.
We analyze the query complexity in \cref{sec-complexity}, showing that our method recovers Grover's optimal quadratic speedup.
Numerical simulations are given in \cref{sec-experiments}.
We conclude in \cref{sec-discussion}.

\section{Preliminary}\label{sec-prelim}

In this section, we begin by reviewing the notation for the unstructured search and introducing the geometric tools required for optimization on the unitary manifold. Afterwards, we summarize the fundamental Riemannian optimization framework that builds upon these ingredients.
For convenience, \cref{tab:notation} collects the basic algebraic notation used throughout the paper.

\begin{table}[!htbp]
\centering
\small
\renewcommand{\arraystretch}{1.05}

\caption{Basic algebraic notation. \(A\) and \(B\) in the table are \(N\times N\) complex matrices.}
\label{tab:notation}

\begin{tabular}{@{}
L{0.18\textwidth}
L{0.75\textwidth}
@{}}
\hline
\textbf{Notation} & \textbf{Meaning} \\
\hline
\(\mathbb{R}\), \(\mathbb{C}\)
& Real field and complex field.\\

\(N\)
& Size of the search space.\\

\(\mathcal{H} \cong \C^N \)
& \(N\)-dimensional Hilbert space. \\

\(\mathbb{C}^{N\times N}\)
& Space of all \(N\times N\) complex matrices, identified with operators on \(\mathcal{H}\).\\

$\bar{z}$, \(\Re z\), \(\Im z\)
& Conjugate, real part and imaginary part of \(z\in\mathbb{C}\). \\

\(\Tr(A)\)
& Trace of \(A\).\\

\(A^\dagger\)
& Conjugate transpose of \(A\). \\

\(\langle A,B\rangle\)
& Real Hilbert--Schmidt inner product,
\(\langle A,B\rangle= \Re \Tr(A^\dagger B)\).\\

\(\|A\|_F\)
& Frobenius norm of \(A\),
\(\|A\|_F=\sqrt{\Tr(A^\dagger A)}\).\\

\(\|A\|_2\)
& Spectral norm of \(A\), i.e., its largest singular value. \\

\(\Skew(A)\)
& Skew-Hermitian part of \(A\),
\(\Skew(A)=\tfrac12(A-A^\dagger)\). \\

\([A,B]\)
& Commutator of \(A\) and \(B\),
\([A,B]=AB-BA\).\\

\(A\otimes B\)
& Tensor product, or Kronecker product, of \(A\) and \(B\).\\

\(\operatorname{span}_{\mathbb{F}}\{\cdot\}\)
& Linear span over \(\mathbb{F}\), where \(\mathbb{F}=\mathbb{R}\) or \(\mathbb{C}\). \\

\([A]_{\mathcal{B}}\)
& Matrix representation of \(A|_{\mathcal{S}}\) with respect to the ordered basis
\(\mathcal{B}=\{b_1,\ldots,b_d\}\) of a subspace
\(\mathcal{S}\subseteq\mathcal{H}\), assuming \(A\mathcal{S}\subseteq\mathcal{S}\);
namely, \([A]_{\mathcal{B}}\in \C^{d \times d}\) satisfies
\(A b_j=\sum_{i=1}^d ([A]_{\mathcal{B}})_{ij} b_i\), \(j=1,\ldots,d\).\\

\hline
\end{tabular}

\end{table}

\subsection{Problem statement}\label{sec-problem}

We consider a search problem over the set $[N]:=\{0,1, \ldots, N-1\}$, where a binary oracle function $g(x) \in\{0,1\}$ identifies the marked items. The marked set is denoted as $S = \{x \in[N]: g (x) = 1\}$ with $ M: = |S|,$ and $1 \leq M \ll N.$
The task of the search problem is to identify at least one marked item.

The quantum version of the search problem is to prepare a superposition of the marked states.
Specifically, let $\mathcal{H}\cong \C^N $ be the $N$-dimensional Hilbert space associated with $n$ qubits and $N=2^n$, and $\{|j\rangle\}_{j=0}^{N-1}$ denote its computational basis (an orthonormal basis) of $\mathcal{H}$ with the $j$-th standard basis vector $|j\rangle$.
Let the marked subspace be $\mathcal{T} = \operatorname{span}_{\mathbb{C}}\{|x\rangle: x \in S\} \subseteq \mathcal{H}$.
Then, the task of the quantum search problem is to prepare a quantum state (i.e., a normalized $N$-dimensional vector) that lies in $\mathcal{T}$.
The projector onto the marked subspace
$\mathcal{T}$ is expressed as
\begin{equation}\label{eq-defn-Hf}
    H_g = \sum_{x \in S} |x\rangle\langle x|.
\end{equation}
Here $\langle x|$ denotes the conjugate transpose of $|x\rangle$.
Notice that $H_g$ is Hermitian and idempotent, i.e., $H_g = H_g^\dagger = H_g^2$.
The initial state of Grover's algorithm is the uniform superposition over all computational basis states,
\begin{equation}\label{eq-defn-psi0}
    |\psi_0\rangle = \frac{1}{\sqrt{N}} \sum_{x=0}^{N-1} |x\rangle,
\end{equation}
and we denote the corresponding rank-one projector by $\psi_0 = |\psi_0\rangle\langle\psi_0|$.
Then, Grover's algorithm is implemented using two types of quantum gates, namely the oracle operator and the diffusion operator, which are defined as follows:
\begin{align}
  U_g (\beta) &= e^{i \beta H_g}
             = I + \bigl (e^{i\beta} - 1\bigr)H_g, \label{eq-defn-Ug} \\
  D (\alpha) &= e^{i \alpha \psi_0}
            = I + \bigl (e^{i\alpha} - 1\bigr)\psi_0, \label{eq-defn-D}
\end{align}
where $\alpha, \beta \in \mathbb{R}$ control the rotation angles. Grover's algorithm alternately applies these two operators through $G (\alpha_k, \beta_k) = - D (\alpha_k) U_g (\beta_k),$ so that after $T$ iterations the final state is given by
\[
|\psi_{\text{final}}\rangle
    = \prod_{k=1}^{T} G (\alpha_k, \beta_k) |\psi_0\rangle,
\]
which is designed to approximate the ideal target state, namely, the uniform superposition over all marked items,
\begin{equation}\label{eq-defn-psistar}
    |\psi^{\star}\rangle
    = \frac{1}{\sqrt{M}} \sum_{x \in S} |x\rangle
    \in \mathcal{T}.
\end{equation}

The query complexity of the Grover's algorithm is typically considered as the number of calls to $U_g(\beta)=e^{i \beta H_g}$.
The circuit implementations of these two types of gates can be found in many standard textbooks on quantum computing \cite{kaye2006introduction,nielsen2010quantum} and are not the focus of this work.

We next place several Grover-type algorithms within the above framework to highlight how different choices of the angle parameters $\alpha_k$ and $\beta_k$ lead to different variants.
In the original formulation of Grover's algorithm \cite{grover1996fast}, the parameters are chosen as $\alpha_k=\beta_k=\pi$ for all $k$.
This choice, however, makes the procedure vulnerable to the soufflé problem \cite{brassard1997searching}, where performing too many Grover iterations drives the state past the target instead of toward it when the number of marked states $M$ is underestimated.
To mitigate this overshooting issue, the $\pi/3$ algorithm \cite{grover2005fixed} uses the fixed angles $\alpha_k=\beta_k=\pi/3$, though at the cost of losing the quadratic speedup.
The fixed-point algorithm \cite{yoder2014fixed} uses a carefully designed sequence of $\alpha_k$ and $\beta_k$ to suppress overshooting while preserving the quadratic speedup.
These Grover-type methods differ mainly in their deterministic choices of angles.
In contrast, the method proposed later in this paper adaptively determines $\alpha_k$ and $\beta_k$ at each iteration using the latest Riemannian gradient and step size information.

\subsection{Geometric tools on \texorpdfstring{$\mathrm{U}(N)$}{U(N)}}

In the next section, we will reformulate the unstructured search problem above as an optimization problem over the compact Lie group of $N\times N$ unitary matrices:
\begin{equation}\label{eq-UN}
    \mathrm{U} (N)=\bigl\{ U\in\C^{N\times N} \big| U^{\dagger}U=I_N\bigr\},
\end{equation}
which is a Riemannian manifold of dimension $N^2$.
Manifold optimization is deeply rooted in the geometry of manifolds. Here, we present only the minimal set of geometric tools required for our interested manifold $\mathrm{U} (N)$.
For a systematic treatment of manifold optimization, the reader is referred to textbooks \cite{absil2008optimization, boumal2023introduction}.

\paragraph{Tangent space}
The tangent space at a point on a manifold represents all possible directions in which one can move infinitesimally from that point.
For example, on the sphere manifold $\mathbb{S}^{n-1}=\left\{x \in \mathbb{R}^n: \|x\|=1\right\}$, the tangent space at $x\in \mathbb{S}^{n-1}$ is $T_x=\left\{v \in \mathbb{R}^n: x^{T} v=0\right\} \subseteq \mathbb{R}^n$, which consists of all directions orthogonal to $x$.
For the unitary manifold, the tangent space at $U \in \mathrm{U} (N)$ is given by
\begin{equation*}
    T_U=\left\{\Omega U: \Omega^{\dagger}=-\Omega\right\}=\mathfrak{u} (N) U,
\end{equation*}
where $\mathfrak{u} (N)=\left\{\Omega \in \mathbb{C}^{N \times N}: \Omega^{\dagger}=-\Omega\right\}$ denotes the Lie algebra of $\mathrm{U} (N)$, i.e., the real vector space of all $N \times N$ skew-Hermitian matrices.
Consider the ambient space $\mathbb{C}^{N \times N}$ equipped with the real inner product $\langle A, B \rangle = \Re \Tr (A^\dagger B). $ With respect to this inner product, the orthogonal projection of an arbitrary $Z \in \mathbb{C}^{N \times N}$ onto $T_U$ is given by $\mathcal{P}_{U} (Z) = \frac12 (Z - U Z^\dagger U) =\Skew (Z U^\dagger)U\in T_U$ where $\Skew (A): = \frac{1}{2} (A - A^\dagger). $
This construction parallels the projection in the sphere manifold, where $\mathcal{P}_x (z)=z-x\left (z^{T} x\right)\in T_x$ removes the component of $z$ along $x$ for any $z \in \mathbb{R}^n$.

\paragraph{Riemannian gradient}
At a point $U$ on the unitary manifold, the Riemannian gradient, $\operatorname{grad} f (U)$, represents the tangent direction of the steepest ascent of a smooth function $f: \mathbb{C}^{N \times N} \rightarrow \mathbb{R}$; it serves as the Riemannian analog of the usual gradient in Euclidean space.
Here, we view $\mathrm{U} (N)$ as a Riemannian submanifold of the real vector space $\mathbb{C}^{N \times N} \cong \mathbb{R}^{2 N^2}$. Each tangent space $T_U$ is endowed with the inner product
$\langle X, Y\rangle_U: =\Re\Tr (X^\dagger Y)=\Tr (\widetilde{X}^\dagger\widetilde{Y})$,
where $X=\widetilde{X} U, Y=\widetilde{Y} U$, and $\widetilde{X}, \widetilde{Y} \in \mathfrak{u} (N)$.
Under this metric, the Riemannian gradient of $f$ is given by (\cite[Proposition 3.61]{boumal2023introduction}):
\begin{equation}\label{eq-gradf-formula}
    \operatorname{grad} f (U)=\mathcal{P}_U (\nabla f (U)) \in T_U,
\end{equation}
where $\nabla f (U)$ is the usual gradient.

\paragraph{Retractions}
Retractions provide a mechanism for moving along tangent directions while ensuring that the iterates remain on the manifold.
For any $U \in \mathrm{U} (N)$, a retraction is a smooth mapping $\mathrm{R}_U: T_U \rightarrow \mathrm{U} (N)$ such that the induced curve $\gamma (t):=\mathrm{R}_U (t \eta)$ satisfies
\begin{equation}\label{eq-defn-retraction}
\gamma(0)=U \text { and } \dot{\gamma}(0)=\eta \text { for all } \eta \in T_U .
\end{equation}
As a simple example, on the sphere manifold a retraction can be implemented by normalization: $\mathrm{R}_x (\eta)=\frac{x+\eta}{\|x+\eta\|}$.
A natural case of retraction on $\mathrm{U} (N)$ is Riemannian exponential map, defined by
\begin{equation}\label{eq-Exp}
    \Exp_U (\eta)=e^{\widetilde{\eta}} U, \qquad \widetilde{\eta}:=\eta U^\dagger\in\mathfrak{u} (N).
\end{equation}
Its induced curve $\gamma (t)=e^{t\widetilde{\eta}}  U$ satisfies $\gamma (0)=U$ and $\dot{\gamma} (0)=\widetilde{\eta}U=\eta$, verifying that $\Exp_U$ is a valid retraction.
Later, in order to ensure a physically implementable retraction, we will restrict the retraction to a specific subspace of  $T_U $, so that the subsequent Riemannian gradients remain within this subspace.

\subsection{Riemannian algorithms}

Recall that in the classical setting, to solve $\max_{x\in\mathbb{R}^n} f (x), $ one constructs a sequence $\{x_k\}\subseteq\mathbb{R}^n$ from an initial iterate $x_0$ via the update rule $x_{k+1} = x_k + t_k \eta_k,$
where $\eta_k\in\mathbb{R}^n$ is the search direction and $t_k>0$ is the step size.
The direction $\eta_k$ can be taken as gradient, or a Newton direction.
The step size may be chosen in various ways: as a fixed constant (e.g., reciprocal of the Lipschitz constant), by backtracking to satisfy the Armijo condition \cite{nocedal2006numerical}, by exact line search \cite{nocedal2006numerical}, or by the Barzilai-Borwein rule \cite{raydan1993barzilai}. Each choice of $\eta_k$ and $t_k$ gives rise to a different optimization algorithm.

This Euclidean update scheme extends naturally to Riemannian manifolds for solving
\begin{equation}\label{eq-max-f}
    \max _{U \in \mathrm{U} (N)} f (U),
\end{equation}
where straight-line updates are replaced by movement along tangent directions, followed by a retraction that maps the iterate back onto the manifold.
Given an initial $U_0\in \mathrm{U} (N)$, the Riemannian update rule constructs a sequence $\{U_k\}\subseteq \mathrm{U} (N)$ via
\begin{equation}\label{eq-update-basic}
    U_{k+1} = \mathrm{R}_{U_k} (t_k \eta_k),
\end{equation}
where $\eta_k\in T_{U_k}$ is the chosen search direction (e.g., Riemannian gradient $\grad f(U_k)$) and $\mathrm{R}_{U_k}: T_{U_k}\to\mathrm{U} (N)$ is a retraction. By choosing appropriate combinations of $\eta_k$ and $t_k$, one recovers the Riemannian analogues of classical algorithms, each inheriting the same convergence guarantees as their Euclidean counterparts.

\section{Grover-compatible Riemannian algorithm}\label{sec-grover-alg}

In this section, we first reformulate the unstructured search problem as a maximization problem over the unitary manifold (as in \cref{eq-max-f}) and establish its global optimality conditions. While one could apply standard Riemannian gradient ascent directly, ensuring that each iteration is physically realizable on a quantum device introduces additional challenges. To address this, we introduce the notion of a Grover-compatible retraction, which enables each update to be both physically implementable and fully consistent with the Riemannian optimization framework. Finally, we find that the resulting iteration process is classically simulable, allowing all parameters required to construct the quantum circuit to be precomputed on a classical computer.

\subsection{Optimization problem on unitary manifold}\label{sec-be-opt}

We first summarize some properties that can be easily observed in \cref{sec-problem}. Applying $H_g$ to the target state yields
$H_g\left|\psi^{\star}\right\rangle=\left|\psi^{\star}\right\rangle$
indicating that $\left|\psi^{\star}\right\rangle$ is an eigenstate of $H_g$ associated with eigenvalue 1. Indeed, the spectrum of $H_g$ consists of eigenvalue 1 on the marked subspace $\mathcal{T}$ with multiplicity $M$, and eigenvalue 0 on its orthogonal complement $\mathcal{T}^{\perp}$ with multiplicity $N-M$. Let
\begin{equation*}
    q_0: =\frac{M}{N} \in (0, 1)
\end{equation*}
denote the probability of selecting a marked item uniformly at random. We find that applying $H_g$ to the initial state gives $H_g\left|\psi_0\right\rangle=\sqrt{q_0}\left|\psi^{\star}\right\rangle$; thus, the initial expectation value satisfies
$\left\langle\psi_0\right| H_g\left|\psi_0\right\rangle=\| H_g\left|\psi_0\right\rangle \|^2=\| \sqrt{q_0}\left|\psi^{\star}\right\rangle \|^2=q_0.$
Indeed, for any state $|\psi\rangle$, the expectation value
$\langle\psi| H_g|\psi\rangle = \sum_{x \in S}|\langle x | \psi\rangle|^2 \in[0, 1]$
represents the probability of observing a marked item when measuring $|\psi\rangle$ in the computational basis. Notice that $\langle\psi^{\star}| H_g|\psi^{\star}\rangle =1$.

From a high-level view, Grover's algorithm (see \cref{sec-problem}) can be interpreted as the process of searching for a quantum state whose expectation value with respect to $H_g$ (thus, the probability of observing a marked item) equals one. To construct such a state, we can adopt the philosophy of quantum circuit design. Starting from an easily prepared state $\left|\psi_0\right\rangle$ as the initial state (e.g., \cref{eq-defn-psi0}), we aim to design a quantum circuit $U$ whose gates are implementable on quantum hardware (e.g., \cref{eq-defn-Ug,eq-defn-D}), such that the resulting state $U\left|\psi_0\right\rangle$ becomes an eigenstate of $H_g$ corresponding to its largest eigenvalue 1.
Consequently, this leads to the following optimization problem.

\begin{problem}[Riemannian optimization formulation]\label{pro-1}
Let $H = H^\dagger = H^2$ be an orthogonal projector on a Hilbert space $\mathcal{H}$, and $M:= \operatorname{rank}(H)$. Fix a unit vector $|\psi_0\rangle$, and define the rank-one projector $\psi_0: = |\psi_0\rangle\langle\psi_0|$.\footnote{In this paper, the kets (e.g., $|\psi\rangle$) denote quantum states, while the corresponding symbols without bra-ket notation (e.g., $\psi$) represent their associated density operators.} Assume that $\left\langle\psi_0\right| H\left|\psi_0\right\rangle \notin\{0, 1\}$.
Consider the optimization problem
\begin{equation}\label{eq-cost}
    \max_{U\in \mathrm{U} (N)} f (U), \quad\quad f(U) = \operatorname{Tr}\left(H U \psi_0 U^\dagger\right)= \langle \psi_0 | U^\dagger H U| \psi_0 \rangle,
\end{equation}
where the feasible region is the compact unitary manifold $\mathrm{U} (N)$ defined in \cref{eq-UN}.
\end{problem}

We now discuss the connection between the above optimization problem and Grover's algorithm. In \cref{pro-1}, the Hermitian observable $H$ is required to satisfy the idempotent condition $H=H^2$, implying that its expectation value $\langle\psi| H|\psi\rangle$ can take values only within the interval $[0, 1]$. The operator $H_g$ introduced in \cref{eq-defn-Hf} is a specific instance of such an observable. The initial state $\left|\psi_0\right\rangle$ in \cref{pro-1} may be any state satisfying $\left\langle\psi_0\right| H\left|\psi_0\right\rangle \notin\{0, 1\}$; equivalently, $H\left|\psi_0\right\rangle \neq 0$ and $H\left|\psi_0\right\rangle \neq\left|\psi_0\right\rangle$.\footnote{Since $H=H^{\dagger}=H^2$, we have $\left\langle\psi_0\right| H\left|\psi_0\right\rangle=\| H\left|\psi_0\right\rangle \|^2 . $
Decompose $\left|\psi_0\right\rangle=H\left|\psi_0\right\rangle+ (I-H)\left|\psi_0\right\rangle$ into orthogonal components. Then $1 = \|H|\psi_0\rangle\|^2 + \| (I - H)|\psi_0\rangle\|^2$.
Hence, $\left\langle\psi_0\right| H\left|\psi_0\right\rangle=1$ if and only if $\| (I-H)\left|\psi_0\right\rangle \|=0$, i.e., $H\left|\psi_0\right\rangle=\left|\psi_0\right\rangle$. }
The uniform state $\left|\psi_0\right\rangle$ defined in \cref{eq-defn-psi0} also constitutes a particular example, as we have shown that $\left\langle\psi_0\right| H_g\left|\psi_0\right\rangle=q_0 \in (0, 1)$.
Nevertheless, the direct optimization formulation in \cref{pro-1} differs from Grover's algorithm in the following two aspects.

\textit{Challenge 1.} Grover's algorithm constrains the circuit $U$ to have the specific structure $U=\prod_{k=1}^T G(\alpha_k,\beta_k)$, reflecting implementability considerations. The oracle operator $U_g(\beta)$ can be realized from a standard Boolean oracle $O_g:|x\rangle|b\rangle\mapsto |x\rangle|b\oplus g(x)\rangle$, where $\oplus$ denotes addition modulo two, by computing $g(x)$ into an ancilla, applying the phase gate $\left(\begin{smallmatrix}1&0\\0&e^{i\beta}\end{smallmatrix}\right)$, and uncomputing $g(x)$. The diffusion operator $D(\alpha)$ is also structured: for the uniform state $|\psi_0\rangle=\frac{1}{\sqrt{N}}\sum_{x=0}^{N-1}|x\rangle=\mathrm{H}^{\otimes n}|0^n\rangle$, where $|0^n\rangle$ is the all-zero computational basis and $\mathrm{H}=\frac{1}{\sqrt{2}}\left(\begin{smallmatrix}1&1\\1&-1\end{smallmatrix}\right)$ is the single-qubit Hadamard gate, we have
\[
D(\alpha)=e^{i\alpha |\psi_0\rangle \langle\psi_0|}
=\mathrm{H}^{\otimes n} (e^{i\alpha |0^n\rangle\langle 0^n|}) \mathrm{H}^{\otimes n}.
\]
Here, $e^{i\alpha|0^n\rangle\langle0^n|} =I+(e^{i\alpha}-1)|0^n\rangle\langle0^n|$ is a multi-controlled phase operation, which has a standard circuit implementation.
In contrast, \cref{pro-1} imposes no structural constraint on $U$. Thus, a feasible point in $\mathrm{U}(N)$ need not be efficiently implementable. For $N=2^n$, a generic unitary is an $N\times N$ matrix with real dimension $N^2=4^n$. Since each one or two-qubit elementary gate contributes only $O(1)$ continuous degrees of freedom, a simple parameter counting argument already indicates an $\Omega(4^n)$ gate lower bound for representing generic $n$-qubit unitaries.
This intuition is made precise in standard results: when arbitrary one-qubit gates and CNOT gates are used as elementary gates, generic $n$-qubit unitaries require at least $\lceil(4^n-3n-1)/4\rceil$ CNOT gates~\cite{shende2005synthesis}. Moreover, some $n$-qubit unitaries require $\Omega(2^n\log(1/\epsilon)/\log n)$ elementary operations to approximate within precision $\epsilon$~\cite[Sec.~4.5.4]{nielsen2010quantum}. Therefore, generic unitaries in $\mathrm{U}(N)$ are generally hard to implement, whereas structured unitaries, such as Grover's oracle and diffusion operators discussed above, may admit circuit implementations of polynomial size.

\textit{Challenge 2.} Grover's algorithm requires the final output state $U\left|\psi_0\right\rangle$ to be the uniform superposition over all marked items, namely the state $\left|\psi^{\star}\right\rangle$ defined in \cref{eq-defn-psistar}. In \cref{pro-1}, by comparison, any unit vector within the 1-eigenspace of $H$ constitutes a valid theoretical solution.

In what follows, we will design an optimization algorithm to solve \cref{pro-1} while preserving the two characteristic features of Grover's algorithm discussed above.
In this way, we successfully combine Grover's algorithm with optimization techniques, thereby enriching it with new perspectives and strengths drawn from the extensive toolbox and theoretical framework of modern optimization.

\subsubsection{Global optimality}

For our cost $f (U)=\Tr(HU\psi_0U^\dagger)$ in \cref{pro-1}, its Euclidean gradient is $\nabla f (U)=2HU\psi_0$. Using identity $\Skew (2 A B)=[A, B]$ for (skew-)Hermitian matrices $A$ and $B$, by \cref{eq-gradf-formula}, we obtain Riemannian gradient:
\begin{equation}\label{eq-grad-defn}
    \operatorname{grad} f (U)=\operatorname{Skew}\left (\nabla f (U) U^{\dagger}\right) U=\left[H, U \psi_0 U^{\dagger}\right] U \in T_U.
\end{equation}
Throughout the paper, we use the tilde symbol to denote the skew-Hermitian component associated with an element of $T_U=\mathfrak{u}(N) U$, e.g., $\widetilde{\operatorname{grad}} f (U)=\left[H, U \psi_0 U^{\dagger}\right]$.
In practice, the rightmost factor $U$ often cancels out, leaving only the skew-Hermitian part as the essential term. Note that $\operatorname{\widetilde{grad}} f (U)=0$ if and only if $\operatorname{grad} f (U)=0$.

Let us now examine the optimality conditions of \cref{pro-1} from the perspective of optimization.
When the manifold constraint for $U$ is ignored, the cost function $f$ reduces to a standard quadratic function in the Euclidean space $\mathbb{C}^{N \times N} \cong \mathbb{C}^{N^2}$. Indeed, it can be written as
\[
f (U)=\operatorname{vec} (U)^{\dagger}\left (\psi_0 \otimes H\right) \operatorname{vec} (U),
\]
where \(\operatorname{vec}(U)\in\mathbb{C}^{N^2}\) denotes the column-stacked vectorization of matrix \(U\), and \(\otimes\) denotes the tensor product.

It is straightforward to verify that $\psi_0 \otimes H$ is positive semidefinite, so $f(U)$ is convex as a quadratic function on the ambient Euclidean space. However, the feasible set $\mathrm{U}(N)$ is nonconvex in $\mathbb{C}^{N\times N}$. Therefore, restricting $f$ to $\mathrm{U}(N)$ yields a nonconvex constrained optimization problem.

Nevertheless, for \cref{pro-1} we can still derive a global optimality condition analogous to the convex case, as stated in \cref{thm-optimality}. In particular, a vanishing Riemannian gradient characterizes global optimality.

\begin{theorem}[Global optimality]\label{thm-optimality}
The skew-Hermitian part of the Riemannian gradient of $f (U)$ in \cref{eq-cost} on $\mathrm{U} (N)$ is given by
\begin{equation*}
    \operatorname{\widetilde{grad}} f (U)=\left[H, \psi_U\right], \quad \psi_U: =U \psi_0 U^{\dagger},
\end{equation*}
where $\psi_U$ denotes the density operator corresponding to the output state of the circuit $U$.
Then, $\operatorname{\widetilde{grad}} f (U^\star)=0$ if and only if $U^\star \in \mathrm{U} (N)$ globally minimizes $f (U)$ over $\mathrm{U} (N)$ with the minimum value 0, or globally maximizes it with the maximum value 1.
\end{theorem}
\begin{proof}
For any Hermitian operator $H$ and any pure state $\psi=|\psi\rangle\langle\psi|$, it holds that $[H, \psi]=0$ if and only if $H|\psi\rangle=\lambda|\psi\rangle$ for some real $\lambda$. Since the eigenvalues of $H$ are $\{0, 1\}$, this condition is equivalent to $H|\psi\rangle=0$ or $H|\psi\rangle=|\psi\rangle$, i.e., $\langle\psi| H|\psi\rangle \in\{0, 1\}$.
\end{proof}

\subsubsection{Naive Riemannian implementation}

Now, if we directly apply update \cref{eq-update-basic} to \cref{pro-1}, we observe that each iteration simply appends new gates to the quantum circuit.
We begin by considering the Riemannian gradient ascent (RGA) method with Exponential map $\mathrm{R}= \Exp$ in \cref{eq-Exp}, that takes the update form
\begin{equation}\label{eq-Exp-updat}
U_{k+1}
=\Exp_{U_k}\left (t_k\grad f (U_k)\right)
= e^{t_k \, \widetilde{\grad}f (U_k)} U_k
= e^{t_k [H, \psi_k]} U_k, \quad k = 0, 1, \dots,
\end{equation}
where $\psi_k: = U_k \psi_0 U_k^\dagger$ denotes the intermediate quantum state after $k$-th iterations.
Setting the initial gate to $U_0 = I$ and the initial uniform state, the resulting state $\lvert \psi_T \rangle = U_T \lvert \psi_0 \rangle$ after $T$ iterations is given by \cref{fig:exp_update}.

\begin{figure}
    \centering
    \includegraphics[width=0.85\linewidth]{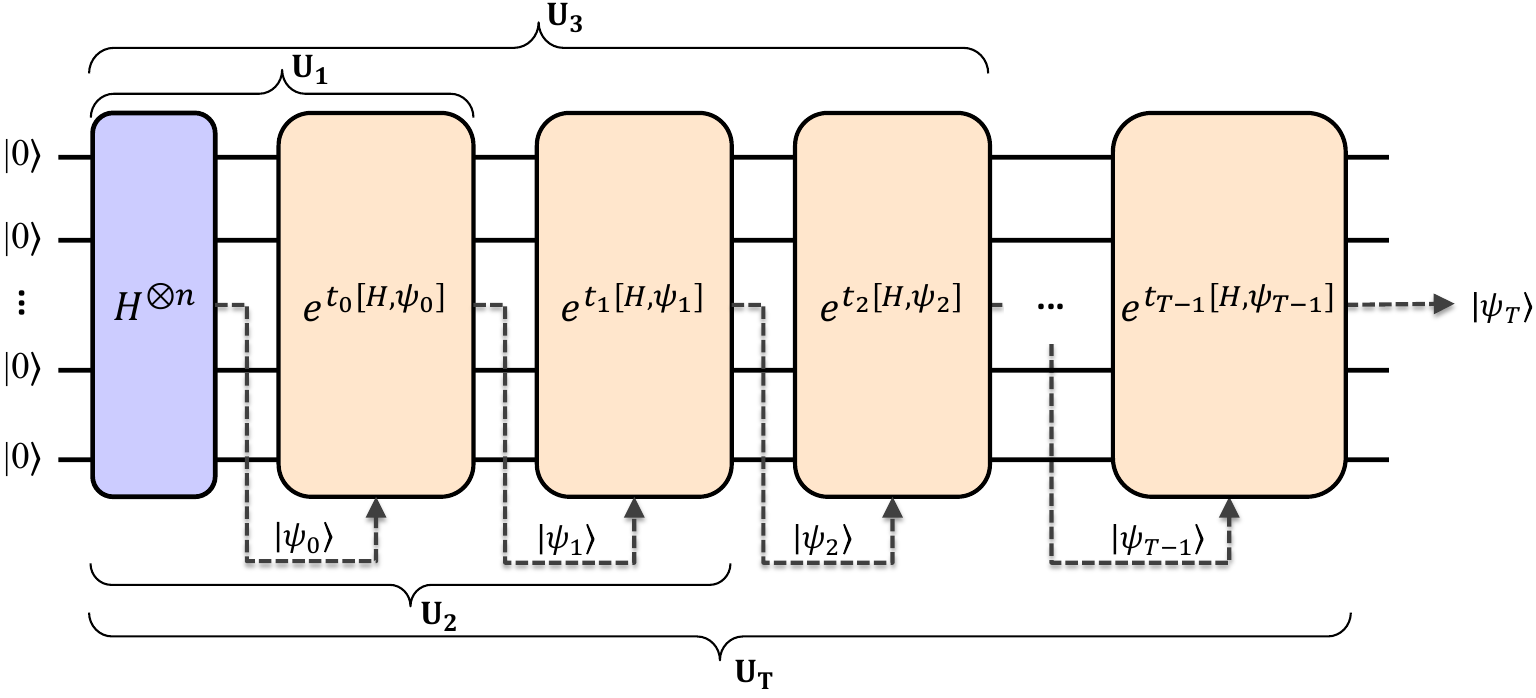}
    \caption{Quantum circuit generated by Riemannian gradient ascent method with the exponential map update \cref{eq-Exp-updat}, starting from $U_0=I$ and the uniform state. At each iteration, a new gate $e^{t_k\left[H, \psi_k\right]}$ is appended, where $\left[H, \psi_k\right]$ is the current Riemannian gradient and $t_k$ is the step size.
    }
    \label{fig:exp_update}
\end{figure}

The above procedure was also proposed in \cite{wiersema2023optimizing} for general quantum circuit design problems.
The overall approach can be summarized as follows: based on the information extracted from the quantum state produced by the current circuit, new gates are added so that the updated circuit output becomes closer to the ground state. This process is then repeated iteratively. Unlike traditional parameterized quantum circuits (PQCs), the ansatz is not fixed but grows dynamically.

Unfortunately, although the above optimization update in \cref{eq-Exp-updat} is theoretically valid, we do not know how to efficiently implement the newly added gate $e^{t_k\left[H, \psi_k\right]}$ in practice. If we repeatedly pause and perform certain measurements or operations on the current quantum state $\left|\psi_k\right\rangle$ to obtain information for approximating $e^{t_k\left[H, \psi_k\right]}$, the procedure becomes inefficient as the complexity would grow exponentially in terms of the number of iterations. Therefore, this naive Riemannian implementation is not working. In the following, we propose a more refined approach.

\subsection{Grover-compatible retractions}

We will introduce the key concept, i.e., Grover-compatible retractions, a class of retractions that not only satisfy the standard defining properties in \cref{eq-defn-retraction} (albeit restricted to a subspace of the tangent space), thereby ensuring convergence from an optimization standpoint, but also admit a physically implementable structure built from the gates $D(\alpha)$ and $U_g(\beta)$.

\subsubsection{Invariant 2D subspace of Riemannian gradients}

Throughout the remainder of the paper, let \(H=H^\dagger=H^2\) be an orthogonal projector on the Hilbert space \(\mathcal H\), let \(|\psi_0\rangle\in\mathcal H\) be a fixed unit vector such that \(0<q_0:=\langle\psi_0|H|\psi_0\rangle<1\), and set \(\psi_0=|\psi_0\rangle\langle\psi_0|\).

To facilitate the analysis of the trajectory of the Riemannian gradient, we first establish two auxiliary lemmas. The first characterizes the norms and orthogonality of the commutators that naturally arise in the computation of the Riemannian gradient.

\begin{lemma}\label{lem-X0YO}
For any state $|\psi\rangle \in \mathcal{H}$, define the rank-one projector $\psi := |\psi\rangle\langle\psi|$ and the scalar $q := \langle\psi|H|\psi\rangle$. Consider the skew-Hermitian operators
\begin{equation*}
    X: =[H, \psi], \qquad Y: = i[H, X] = i[H, [H, \psi]].
\end{equation*}
Then, it holds that $\|X\|_F = \|Y\|_F=\sqrt{2 q (1-q)}$ and $\langle X, Y \rangle = 0$.
\end{lemma}

\begin{proof}
We begin with a useful identity. Since $H^2=H$, for any matrix $A$,
\begin{equation}\label{eq-411}
    [H,[H,[H, A]]]=[H, A] .
\end{equation}
Applying \cref{eq-411} with $A=\psi$ gives $[H, [H, X]]=X$. Then,
\begin{align*}
    \|Y\|_F^2
    =\operatorname{Tr}\big ([H, X]^\dagger[H, X]\big)
   =\operatorname{Tr}\big (X^\dagger[H, [H, X]]\big)
    =\operatorname{Tr} (X^\dagger X)
    =\|X\|_F^2.
\end{align*}
For the inner product, note that
\begin{equation*}
    \left\langle X, Y\right\rangle
    =-i\operatorname{Tr}\left (X[H, X]\right)
    =i\operatorname{Tr}\left (X[X, H]\right)
    =i\operatorname{Tr}\left ([X, X] H\right)=0.
\end{equation*}
To evaluate $\|X\|_F^2$, we expand $X^\dagger X=\psi H\psi - \psi H \psi H - H\psi H\psi + H\psi H.$ Taking the trace and using cyclicity gives $\|X\|_F^2 =\operatorname{Tr} (H\psi)+\operatorname{Tr} (H\psi) -2 \operatorname{Tr} (|\psi\rangle\langle\psi|H|\psi\rangle\langle\psi|H) = 2 q (1-q).$ This completes the proof.
\end{proof}

A key geometric feature of Grover-type algorithms \cite{grover1996fast,grover2005fixed,yoder2014fixed} is that all generated states $|\psi_k\rangle$ remain confined to a fixed two-dimensional subspace $\mathcal{S} \subseteq \mathcal{H}$, commonly referred to as the \textit{Grover plane} \cite{miyake2001geometric}. We now define this invariant subspace $\mathcal{S}$ and summarize the action of the relevant operators on $\mathcal{S}$ and its orthogonal complement.

\begin{lemma}[Grover plane]\label{lem-subspace-inv}
Define the Grover plane as the complex subspace
\begin{equation}\label{eq-grover-plane}
\mathcal{S} := \operatorname{span}_{\mathbb{C}} \{ |\psi_0\rangle, H|\psi_0\rangle \} \subseteq \mathcal{H},
\end{equation}
and let $\mathcal{S}^\perp$ be its orthogonal complement.
For any state $|\psi\rangle \in \mathcal{S}$, define the operators $\psi := |\psi\rangle\langle\psi|$, $X := [H,\psi]$ and $Y := i[H,X]$. Then the following properties hold:
\begin{enumerate}[label=(\alph*)]
    \item
    The actions of $H$, $I-H$, and operators derived from $\psi$ on the subspaces $\mathcal{S}$ and $\mathcal{S}^\perp$ are summarized as follows:

    \vspace{1mm}
    \begin{center}
    \renewcommand{\arraystretch}{1.25}
    \begin{tabular}{lll}
    \hline
    {Operator} ($T: \mathcal{H} \to \mathcal{H}$)& {Action on $\mathcal{S}$} & {Action on $\mathcal{S}^\perp$} \\
    \hline
    $H, \ I-H$ & Preserves ($T\mathcal{S} \subseteq \mathcal{S}$) & Preserves ($T\mathcal{S}^\perp \subseteq \mathcal{S}^\perp$) \\

    $\psi, \ X, \ Y$ & Preserves ($T\mathcal{S} \subseteq \mathcal{S}$) & Vanishes ($T\mathcal{S}^\perp = \{0\}$) \\

    $e^{i\theta H}$
    & Preserves ($T\mathcal{S} \subseteq \mathcal{S}$)
    & Preserves ($T\mathcal{S}^\perp \subseteq \mathcal{S}^\perp$) \\

    $e^{i\theta \psi}, \ e^{\theta X}$
    & Preserves ($T\mathcal{S} \subseteq \mathcal{S}$)
    & Identity ($T|_{\mathcal{S}^\perp}=I_{\mathcal{S}^\perp}$) \\
    \hline
    \end{tabular}
    \end{center}
    \vspace{1mm}

    \item For any two linear operators $T_1, T_2: \mathcal{H} \to \mathcal{H}$, if they both vanish on the orthogonal complement $\mathcal{S}^\perp$ (i.e., $T_1|_{\mathcal{S}^\perp} = T_2|_{\mathcal{S}^\perp} = 0$), then $T_1 = T_2$ on the entire space $\mathcal{H}$ if and only if their restrictions to $\mathcal{S}$ are equal ($T_1|_{\mathcal{S}} = T_2|_{\mathcal{S}}$).

\end{enumerate}
\end{lemma}

\begin{proof}
For property (a), we first verify the actions on $\mathcal{S}$. By definition, $H|\psi_0\rangle \in \mathcal{S}$, and since $H^2 = H$, $H(H|\psi_0\rangle) \in \mathcal{S}$. Therefore, $H\mathcal{S}\subseteq \mathcal{S}$. Since the identity operator $I$ trivially preserves $\mathcal{S}$, it follows that $(I-H)\mathcal{S} \subseteq \mathcal{S}$. For any $|\psi\rangle \in \mathcal{S}$, the rank-one projector $\psi$ maps any vector to $\operatorname{span}_{\mathbb{C}}\{|\psi\rangle\} \subseteq \mathcal{S}$, so $\psi \mathcal{S} \subseteq \mathcal{S}$. It follows that the commutators $X$ and $Y$ also preserve $\mathcal{S}$. Operator exponentials $e^{i\theta H}$, $e^{i\theta \psi}$ and $e^{\theta X}$ preserve $\mathcal{S}$ because every term in their Taylor expansions preserves $\mathcal{S}$.
Regarding the actions on $\mathcal{S}^\perp$ in property (a), since $H$ is Hermitian and preserves $\mathcal{S}$, it must preserve its orthogonal complement $\mathcal{S}^\perp$, and so does $I-H$. Consequently, $e^{i\theta H}$ also preserves $\mathcal{S}^\perp$ by its Taylor expansion. For any $|w\rangle \in \mathcal{S}^\perp$, since $|\psi\rangle \in \mathcal{S}$, we have $\langle \psi | w \rangle = 0$, which yields $\psi |w\rangle = 0$. Consequently, $X |w\rangle = (H\psi - \psi H)|w\rangle = - \psi H|w\rangle = 0,$ where we used $H|w\rangle \in \mathcal{S}^\perp$. Similarly, $Y |w\rangle = 0$. Thus $\psi$, $X$ and $Y$ vanish on $\mathcal{S}^\perp$. Moreover, since $\psi|_{\mathcal{S}^\perp}=0$ and $X|_{\mathcal{S}^\perp}=0$, their exponential maps satisfy $e^{i\theta\psi}|_{\mathcal{S}^\perp}=I_{\mathcal{S}^\perp},$ $e^{\theta X}|_{\mathcal{S}^\perp}=I_{\mathcal{S}^\perp}.$ This proves property (a).
Finally, property (b) is a direct consequence of the orthogonal decomposition $\mathcal{H} = \mathcal{S} \oplus \mathcal{S}^\perp$.
\end{proof}

The following theorem establishes an important equivalence: for any non-optimal current state $|\psi_k\rangle=U_k|\psi_0\rangle$, belonging to the Grover plane is equivalent to having its skew-Hermitian gradient component $[H,\psi_k]$ lie in a fixed two-dimensional real subspace $\mathcal{W}\subseteq\mathfrak u(N)$.

\begin{theorem}[2D dynamics of states and gradients]\label{thm-2d-dynamics}
Let $\mathcal{S}$ be the Grover plane defined in \cref{eq-grover-plane}, and assume that $0<q_0:=\langle\psi_0|H|\psi_0\rangle<1.$ Define
\begin{equation}\label{eq-W}
    X_0 := [H,\psi_0],
    \qquad
    Y_0 := i[H,X_0],
    \qquad
    \mathcal{W}:=\operatorname{span}_{\mathbb{R}}\{X_0,Y_0\}\subseteq \mathfrak{u}(N).
\end{equation}
Then $\mathcal W$ is a fixed two-dimensional real subspace of $\mathfrak u(N)$.
Moreover, for any pure state $\psi=|\psi\rangle\langle\psi|$, the following claims hold:
\begin{enumerate}[label=(\alph*)]
    \item If $|\psi\rangle\in\mathcal{S}$, then $[H,\psi]\in\mathcal{W}$.
    \item If $0<q:=\langle\psi|H|\psi\rangle<1$ and $[H,\psi]\in\mathcal{W}$, then $|\psi\rangle\in\mathcal{S}$.
\end{enumerate}
\end{theorem}

\begin{proof}
Since \(q_0\in(0,1)\), \cref{lem-X0YO} gives $\|X_0\|_F=\|Y_0\|_F=\sqrt{2q_0(1-q_0)}\neq0,$ and $\langle X_0,Y_0\rangle=0.$ Hence \(\mathcal W\) is a well-defined two-dimensional real subspace.
To prove the claims (a) and (b), let us first introduce an orthogonal, but not normalized, basis $\mathcal{B} := \{u, v\}$ for Grover plane $\mathcal{S}$, where
\begin{equation}\label{eq-1125}
u := H|\psi_0\rangle, \quad v := (I - H)|\psi_0\rangle, \quad \text{such that} \quad |\psi_0\rangle = u + v.
\end{equation}
We have $\|u\|^2 = \langle \psi_0 | H | \psi_0 \rangle = q_0$ and $\|v\|^2 = 1 - q_0$. Moreover,
\[
    \mathcal S
    =
    \operatorname{span}_{\mathbb C}\{|\psi_0\rangle,H|\psi_0\rangle\}
    =
    \operatorname{span}_{\mathbb C}\{u,v\}.
\]
When representing an operator restricted to $\mathcal{S}$ in the basis $\mathcal{B}$, we observe that $Hu=u$ and $Hv=0$, yielding $[H]_{\mathcal{B}}=\left(\begin{smallmatrix}1&0\\0&0\end{smallmatrix}\right)$. The rank-one projector $\psi_0$ acts as $\psi_0 u=q_0(u+v)$ and $\psi_0 v=(1-q_0)(u+v)$, hence $[\psi_0]_{\mathcal{B}}=\left(\begin{smallmatrix}q_0&1-q_0\\q_0&1-q_0\end{smallmatrix}\right)$.
It then follows that
\begin{gather*}
[X_0]_{\mathcal{B}} = [H, \psi_0]_{\mathcal{B}} = [H]_{\mathcal{B}}[\psi_0]_{\mathcal{B}} - [\psi_0]_{\mathcal{B}}[H]_{\mathcal{B}} =
\begin{pmatrix}
0 & 1-q_0 \\
-q_0 & 0
\end{pmatrix},\\
[Y_0]_{\mathcal{B}} = i[H, X_0]_{\mathcal{B}} = i[H]_{\mathcal{B}}[X_0]_{\mathcal{B}} - i[X_0]_{\mathcal{B}}[H]_{\mathcal{B}} =
i\begin{pmatrix}
0 & 1-q_0 \\
q_0 & 0
\end{pmatrix}.
\end{gather*}

We now prove claim (a). Suppose $|\psi\rangle\in\mathcal S$. Then there exist $\alpha,\beta\in\mathbb C$ such that $|\psi\rangle=\alpha u+\beta v.$
The normalization implies $|\alpha|^2 q_0 + |\beta|^2(1-q_0) = 1$. Therefore,
\begin{equation}\label{eq-1117}
q := \langle\psi| H |\psi\rangle = (\bar{\alpha} u^\dagger + \bar{\beta} v^\dagger) H (\alpha u + \beta v) = |\alpha|^2 q_0,
\end{equation}
which immediately yields $1 - q = |\beta|^2(1-q_0)$.
To find the matrix representation of operator $\psi = |\psi\rangle\langle\psi|$ in $\mathcal{B}$, we consider its action on the basis vectors. Using $ u^\dagger u = q_0$ and $v^\dagger v = 1-q_0$, we obtain $\psi u = \bar{\alpha} q_0 |\psi\rangle$ and $\psi v = \bar{\beta} (1-q_0) |\psi\rangle$. Substituting $|\psi\rangle = \alpha u + \beta v$ into these expressions directly yields
\begin{equation*}
[\psi]_{\mathcal{B}} =
\begin{pmatrix}
|\alpha|^2 q_0 & \alpha \bar{\beta} (1-q_0) \\
\bar{\alpha} \beta q_0 & |\beta|^2(1-q_0)
\end{pmatrix}
=
\begin{pmatrix}
q & (1-q_0) z \\
q_0 \bar{z} & 1-q
\end{pmatrix},
\end{equation*}
where we define $z := \alpha \bar{\beta}$.
With $[\psi]_{\mathcal{B}}$ and $[H]_{\mathcal{B}}$ established, we can concisely compute
\begin{align}
[H, \psi]_{\mathcal{B}}
&= [H]_{\mathcal{B}}[\psi]_{\mathcal{B}} - [\psi]_{\mathcal{B}}[H]_{\mathcal{B}} \notag \\
&=
\begin{pmatrix}
0 & (1-q_0) z \\
-q_0 \bar{z} & 0
\end{pmatrix} \notag  \\
&= \Re(z)
\begin{pmatrix}
0 & 1-q_0 \\
-q_0 & 0
\end{pmatrix}
+ \Im(z)
\begin{pmatrix}
0 & i(1-q_0) \\
i q_0 & 0
\end{pmatrix} \notag  \\
&= \Re (z) [X_0]_{\mathcal{B}} + \Im (z) [Y_0]_{\mathcal{B}}. \label{eq-1123}
\end{align}
Thus, we have established the operator identity $[H, \psi]|_{\mathcal{S}} = (\Re(z) X_0 + \Im(z) Y_0)|_{\mathcal{S}}$.

To extend this equivalence to the entire Hilbert space $\mathcal{H}$, we examine the action of these operators on $\mathcal{S}^\perp$. Since $|\psi_0\rangle, |\psi\rangle \in \mathcal{S}$, applying \cref{lem-subspace-inv} (a) ensures that the operators $[H, \psi]$, $X_0 $, and $Y_0$ all vanish on $\mathcal{S}^\perp$. Therefore, according to \cref{lem-subspace-inv} (b), the operators are globally identical:
\begin{equation*}
[H, \psi] = \Re (z) X_0 + \Im (z) Y_0 \in \mathcal{W}.
\end{equation*}
This completes the proof that $[H, \psi] \in \mathcal{W}$.

We next prove claim (b). Let \(a:=H|\psi\rangle\) and \(b:=(I-H)|\psi\rangle\), so that \(|\psi\rangle=a+b\). The assumption \(0<q=\langle\psi|H|\psi\rangle<1\) implies that \(a\) and \(b\) are nonzero. Moreover, $[H,\psi] = |a\rangle\langle b| - |b\rangle\langle a|.$
On the other hand, from \eqref{eq-1125} we have
\[
X_0 = |u\rangle\langle v| - |v\rangle\langle u|,
\qquad
 Y_0 = i|u\rangle\langle v| + i|v\rangle\langle u|.
\]
Hence, we can rewrite $\mathcal W = \{xX_0 + y Y_0:x,y\in\mathbb R\}$ as $\mathcal W = \left\{ z|u\rangle\langle v| - \bar z |v\rangle\langle u| : z\in\mathbb C \right\}.$ Since $[H,\psi]\in\mathcal W$, there exists $z\in\mathbb C$ such that
\[
    |a\rangle\langle b|
    -
    |b\rangle\langle a|
    =
    z|u\rangle\langle v|
    -
    \bar z |v\rangle\langle u|.
\]
Taking the $H(\cdot)(I-H)$ block of both sides gives $|a\rangle\langle b| = z|u\rangle\langle v|.$ Since \(a,b,u,v\) are all nonzero, the equality of these rank-one operators implies \(a\in\operatorname{span}_{\mathbb C}\{u\}\) and \(b\in\operatorname{span}_{\mathbb C}\{v\}\). Consequently, $|\psi\rangle=a+b\in\operatorname{span}_{\mathbb C}\{u,v\}=\mathcal S.$ This proves claim (b), and completes the proof.
\end{proof}

We next identify an explicit class of updates that preserves the 2D dynamics. The following theorem shows that if the state is evolved using only the two types of Grover gates, $D(\alpha)$ and $U_g(\beta)$, regardless of the choices of the angle parameters, then all states remain in the Grover plane. Consequently, the skew-Hermitian part of Riemannian gradient remains in the fixed subspace $\mathcal{W}$.

\begin{theorem}[Grover-type updates preserve 2D dynamics]\label{thm-grad-in-W}
Initialize the unitary as $U_0=I$. For $k=0,1, \ldots$, consider the update rule:
\begin{equation*}
V_k := \text{a finite product of the form } e^{i \theta H} \text{ and } e^{i \theta \psi_0}, \quad U_{k+1} := V_k U_k,
\end{equation*}
and $\psi_k := |\psi_k\rangle\langle\psi_k|$ with $|\psi_k\rangle = U_k |\psi_0\rangle$; i.e., $|\psi_{k+1}\rangle=V_k|\psi_k\rangle$. Define the scalars $q_k:=\langle\psi_k|H|\psi_k\rangle \in[0, 1]$.
Let $\mathcal{S}$ be the Grover plane defined in \cref{eq-grover-plane}, and $\mathcal{W}$ be the gradient subspace defined in \cref{eq-W}.
Then, for all $k \geq 0$:
\begin{enumerate}[label=(\arabic*)]
\item the state $|\psi_k\rangle$ remains in $\mathcal{S}$,
\item the skew-Hermitian part of Riemannian gradient, $[H, \psi_k]$, remains in $\mathcal{W}$ and $\|[H,\psi_k]\|_F=\sqrt{2q_k(1-q_k)}$.
\end{enumerate}
\end{theorem}

\begin{proof}
We first prove claim (1) by induction. Since $|\psi_0\rangle\in\mathcal{S}$ by definition, the base case holds. By \cref{lem-subspace-inv} (a), both $e^{i\theta H}$ and $e^{i\theta\psi_0}$ preserve $\mathcal{S}$. Hence any finite product $V_k$ of these operators also preserves $\mathcal{S}$. Therefore, if $|\psi_k\rangle\in\mathcal{S}$, then $|\psi_{k+1}\rangle=V_k|\psi_k\rangle\in\mathcal{S}.$ This completes the induction and proves claim (1). For claim (2), since $|\psi_k\rangle\in\mathcal{S}$ for all $k$, \cref{thm-2d-dynamics} (a) implies that $[H,\psi_k]\in\mathcal{W}.$ For the Frobenius norm $\|[H,\psi_k]\|_F=\sqrt{2q_k(1-q_k)}$, see \cref{lem-X0YO}.
\end{proof}

\begin{remark}
The RGA update with the Exponential map in \cref{eq-Exp-updat} also preserves the 2D dynamics by taking $V_k=e^{t_k[H,\psi_k]}$ in \cref{thm-grad-in-W}. Starting from $k=0$, we have $|\psi_0\rangle\in\mathcal{S}$ by definition. If $|\psi_k\rangle\in\mathcal{S}$, then \cref{lem-subspace-inv} (a) implies that $[H,\psi_k]$ preserves $\mathcal{S}$, and hence so does $V_k$. Therefore $|\psi_{k+1}\rangle = V_k|\psi_k\rangle \in \mathcal{S}$, proving by induction that all iterates remain in the Grover plane. However, as discussed above, implementing $e^{t_k[H,\psi_k]}$ as a quantum circuit is not straightforward. We therefore do not pursue this update further.
\end{remark}

\subsubsection{Grover-compatible retractions}

In what follows, whenever it causes no ambiguity, the term \textit{gradient} will also refer to the skew-Hermitian component of the Riemannian gradient. Likewise, when discussing the subspace $\mathcal{W} U \subseteq  \mathfrak{u} (N) U = T_U$, we sometimes, for simplicity, focus on $\mathcal{W} \subseteq \mathfrak{u}(N)$ itself and disregard the right $U$. Recall the notations
\begin{equation*}
    X_0 :=\left[H, \psi_0\right], \quad Y_0:=i\left[H, X_0\right], \quad \mathcal{W} :=\operatorname{span}_{\mathbb{R}}\left\{X_0, Y_0\right\} \subseteq \mathfrak{u} (N).
\end{equation*}

Our key idea is to restrict the optimization dynamics to the fixed 2D subspace $\mathcal{W}$.
At the initial step $k=0$, the gradient $\left[H, \psi_0\right]$ belongs to $\mathcal{W}$.
So we assume that, at some iteration $k$, the current gradient $\left[H, \psi_k\right]$ lies in $\mathcal{W}$. Next, we may employ an elaborate retraction, termed the Grover-compatible retraction $\mathrm{R}_U$ below, to update.

\begin{definition}\label{defn-grover-retraction}
A mapping $\mathrm{R}_U: \mathcal{W} U \rightarrow \mathrm{U}(N)$ is called a \emph{Grover-compatible retraction} if it satisfies the following two conditions:
\begin{enumerate}[label=(\arabic*)]
    \item For all $U \in \mathrm{U}(N)$, \(\mathrm{R}_U\) is a valid retraction on the restricted tangent subspace \(\mathcal WU\). Namely, for every \(\eta\in \mathcal WU\), the induced curve \(\gamma(t):=\mathrm{R}_U(t\eta)\), \(t\geq 0\), satisfies $\gamma(0)=U$ and $\dot{\gamma}(0)=\eta$.

    \item  For any tangent vector $\eta \in \mathcal{W} U \cong \mathbb{R}^2$, let $(x, y) \in \mathbb{R}^2$ be the unique coefficients such that its skew-Hermitian part is given by $\widetilde{\eta}:=\eta U^{\dagger} = x X_0 + y Y_0$. Then, the retraction must take the form
\begin{equation}\label{eq-grover-form}
        \mathrm{R}_{U}(t\eta) = V(t; x, y) U,
        \qquad
        V(t; x, y) = \prod_{\ell=1}^{K} e^{i \theta_\ell^{(1)}(t; x, y) H} e^{i \theta_\ell^{(2)}(t; x, y) \psi_0},
\end{equation}
for some finite integer \(K\) and real-valued parameter functions
\(\theta_\ell^{(1)}\) and \(\theta_\ell^{(2)}\). Here, $V(t; x, y)$ is a unitary operator (representing the newly added gates) constructed as a finite product of gates generated by $H$ and $\psi_0$. Equivalently, we can set $t=1$ above to obtain the expression for $\mathrm{R}_U(\eta)$.
\end{enumerate}
\end{definition}

\begin{remark}
Since the step sizes in our algorithms are always positive, we may safely restrict the domain of the curve $\gamma(t)$ to $t \geq 0$. This modification has no impact on the subsequent theoretical analysis but helps avoid certain technical subtleties. In this setting, $\dot{\gamma} (0)$ should be understood as the right-hand derivative.
\end{remark}

Given a Grover-compatible retraction $\mathrm{R}_U$, we naturally choose the tangent direction $\eta_k$ to be the current gradient. Then, the update $U_{k+1} = \mathrm{R}_{U_k}\left (t_k \eta_k\right)$ is well-defined, since the current gradient $\widetilde{\eta}_k=\left[H, \psi_k\right]$ lies within $\mathcal{W}$.
By construction, the additional circuit components $V$ (in \cref{eq-grover-form}) introduced by $\mathrm{R}_U$ consist solely of gates of the form $e^{i \theta H}$ and $e^{i \theta \psi_0}$.
According to \cref{thm-grad-in-W}, such an update preserves the gradient subspace invariance, ensuring that the next gradient $\left[H, \psi_{k+1}\right]$ also remains in $\mathcal{W}$.

Repeating this process iteratively, the Riemannian gradient ascent (RGA) method based on the Grover-compatible $\mathrm{R}_U$ exactly follows the iterative structure described in \cref{thm-grad-in-W}, ensuring that each gradient $\left[H, \psi_k\right]$ always stays within $\mathcal{W}$. Finally, since the mapping $\mathrm{R}_U$ satisfies the axioms of a valid retraction (the condition (1) in \cref{defn-grover-retraction}; though restricted to $\mathcal{W}$), the standard convergence guarantees of Riemannian optimization remain applicable.

Indeed, there exist numerous retractions that fulfill \cref{defn-grover-retraction}, with varying structures and lengths. Three representative examples are given in the next proposition. It is worth noting that while we label them as 8-factor, 6-factor, and 5-factor retractions, there are multiple possible realizations for each given length.

\begin{proposition}\label{prop-grover-retr}
For all $U \in \mathrm{U} (N)$ and any $\eta \in \mathcal{W} U$, decompose the skew-Hermitian part $\widetilde{\eta}: =\eta U^{\dagger} \in \mathcal{W}$ as $\widetilde{\eta}=x X_0+y Y_0$ for unique coefficients $ (x, y) \in \mathbb{R}^2$. The following are some Grover-compatible retractions as in \cref{defn-grover-retraction}.

\begin{itemize}
    \item  \textbf{8-factor retraction.} Define
\begin{equation*}
    \mathrm{R}^{(8)}_U (\eta):=e^{i \frac{y}{2}  \psi_0} e^{i \frac{\pi}{2} H} e^{-i \frac{x}{2}  \psi_0} e^{-i \pi H} e^{i \frac{x}{2}  \psi_0} e^{-i \frac{\pi}{2} H} e^{-i \frac{y}{2}  \psi_0} e^{i \pi H} \, U.
\end{equation*}
We can show that $\gamma (t)=\mathrm{R}^{(8)}_U (t\eta)$ satisfies $\gamma (0)=U$ and $\dot{\gamma} (0)=\eta \in \mathcal{W} U$.
    \item \textbf{6-factor retraction.} Let $c_1: =-\tfrac{x+y}{2},$ and $c_2: =\tfrac{x-y}{2}$. Define
\begin{equation*}
    \mathrm{R}^{(6)}_U (\eta): =e^{i \frac{\pi}{2} H} e^{i c_1 \psi_0} e^{-i \pi H} e^{i c_2 \psi_0} e^{i \frac{\pi}{2} H} e^{i y \psi_0} \, U.
\end{equation*}
We can show that $\gamma (t)=\mathrm{R}^{(6)}_U (t\eta)$ satisfies $\gamma (0)=U$ and $\dot{\gamma} (0)=\eta \in \mathcal{W} U$.
\item \textbf{5-factor retraction.} Let $A: =\operatorname{atan2} (y, x), R: =\sqrt{x^2+y^2},$ (the argument and modulus of the complex number $x+i y$) and set
\begin{equation}\label{eq-5factor-paras}
    a_1=A+\frac{\pi}{2},  \quad a_2=A-\frac{\pi}{2},  \quad b_1=-\frac{R}{2},   \quad b_2= \frac{R}{2}.
\end{equation}
Define
\begin{equation}\label{eq-5factor-retraction}
    \mathrm{R}^{(5)}_U (\eta): =e^{i a_1 H} e^{i b_1 \psi_0} e^{i\left (a_2-a_1\right) H} e^{i b_2 \psi_0} e^{-i a_2 H} \, U.
\end{equation}
Since $a_2-a_1=-\pi$, the middle $H$-factor simplifies to $e^{-i \pi H}$. Then, for $t \geq 0$,
\begin{gather}
\gamma (t) =
\mathrm{R}^{(5)}_U (t\eta)=V(t )\,  U,  \label{eq-5factor-gamma} \\
V(t )=e^{i a_1 H} e^{i t b_1 \psi_0} e^{i\left(a_2-a_1\right) H} e^{i t b_2 \psi_0} e^{-i a_2 H}.  \label{eq-5factor-gamma-V}
\end{gather}
We can show that $\gamma (t)$ satisfies $\gamma (0)=U$ and $\dot{\gamma} (0)=\eta \in \mathcal{W} U$.
\end{itemize}
\end{proposition}
\begin{proof}
We provide the proof only for the 5-factor retraction, as the arguments for the other cases are similar. Let $\operatorname{ad}_H$ be the linear map on matrices, $\operatorname{ad}_H (A): =[H, A].$
The identity $[H, [H, [H, A]]]=[H, A]$ in \cref{eq-411} implies $\operatorname{ad}_H^3=\operatorname{ad}_H$. Then, any power of $\operatorname{ad}_H$ reduces to one of the three operators $I, \operatorname{ad}_H$, or $\operatorname{ad}_H^2$. In fact, by applying the adjoint action identity in \cite[Proposition 3.35]{hall2015liegroups} (the second equality below), we obtain
\begin{align*}
\Psi (\theta): =e^{i \theta H} \psi_0 e^{-i \theta H}=\sum_{n \geq 0} \frac{ (i \theta)^n}{n! } \operatorname{ad}_H^n (\psi_0)
 =\psi_0+i \sin \theta X_0+ (\cos \theta-1)\left[H, \left[H, \psi_0\right]\right].
\end{align*}
Since $Y_0: =i\left[H, X_0\right]=i\left[H, \left[H, \psi_0\right]\right]$, we have $\left[H, \left[H, \psi_0\right]\right]=-i Y_0$. Hence
\begin{equation}\label{eq-adj}
    i \Psi (\theta)=i \psi_0-\sin \theta X_0+ (\cos \theta-1) Y_0.
\end{equation}
For any $U \in \mathrm{U} (N)$ and $\eta \in \mathcal{W} U$, write the skew-Hermitian part $\widetilde{\eta}: =\eta U^{\dagger} \in \mathcal{W}$ as $\widetilde{\eta}=x X_0+y Y_0$.
Consider the induced curve $\gamma (t)$ in \cref{eq-5factor-gamma} with parameters given in \cref{eq-5factor-paras}. At $t=0$, $\gamma (0)=e^{i a_1 H} e^{i (a_2-a_1) H} e^{-i a_2 H} U=e^{i 0 \cdot H} U=U.$
Recall from \cref{defn-grover-retraction} the definition of the newly added gates $V$, specifically given by \cref{eq-5factor-gamma-V} for the 5-factor retraction. Because only the matrix exponentials involving $\psi_0$ depend on $t$, the product rule yields
\begin{equation*}
    \begin{aligned}
    \dot{V} (0)
    & =i b_1 \, e^{i a_1 H} \psi_0 e^{-i\left (\pi+a_2\right) H}+i b_2 \, e^{i\left (a_1-\pi\right) H} \psi_0 e^{-i a_2 H} \\
    & =i b_1\, \underbrace{e^{i a_1 H} \psi_0 e^{-i a_1 H}}_{\Psi\left (a_1\right): =}+i b_2 \, \underbrace{e^{i a_2 H} \psi_0 e^{-i a_2 H}}_{\Psi\left (a_2\right): =},
    \end{aligned}
\end{equation*}
where we used $a_1-\pi=a_2$ and $\pi+a_2=a_1$. Therefore
\begin{equation*}
    \begin{aligned}
    & \dot{V} (0) =i b_1 \Psi\left (a_1\right)+i b_2 \Psi\left (a_2\right) \\
    & =i\left (b_1+b_2\right) \psi_0-\left (b_1 \sin a_1+b_2 \sin a_2\right) X_0+\left (b_1\left (\cos a_1-1\right)+b_2\left (\cos a_2-1\right)\right) Y_0.
    \end{aligned}
\end{equation*}
With $b_1=-R / 2$ and $b_2=R / 2$, we have $b_1+b_2=0$, so the $i \psi_0$ term vanishes. Using $a_1=A+\frac{\pi}{2}, a_2=A-\frac{\pi}{2}$, then $\sin a_1=\cos A, \sin a_2=-\cos A, \cos a_1=-\sin A, \cos a_2=\sin A$. Hence, $-\left (b_1 \sin a_1+b_2 \sin a_2\right) =R \cos A,$ and $b_1\left (\cos a_1-1\right)+b_2\left (\cos a_2-1\right)=R \sin A.$
Therefore, $\dot{V} (0)= (R \cos A) X_0+ (R \sin A) Y_0=x X_0+y Y_0=\widetilde{\eta}.$
Finally, $\dot{\gamma} (0)=\dot{V} (0) U=\widetilde{\eta} U=\eta .$
\end{proof}

For a general method of constructing Grover-compatible retractions, please refer to Supplementary Material SM1. \cref{tab:retractions-summary} summarizes all the retractions discussed, detailing their domains, Grover compatibility, and whether they preserve the Grover plane. A detailed discussion on whether the Cayley, Polar, and QR retractions preserve the Grover plane is provided in Supplementary Material SM2 and SM3.

\begin{remark}\label{rem:retraction-tradeoff}
While the $5$-factor retraction minimizes per-iteration circuit depth for NISQ hardware, longer retractions ($6$- or $8$-factor) can potentially improve the optimization landscape, such as by yielding a smaller effective Riemannian (pullback) Lipschitz constant. Although rigorously establishing this theoretical relationship remains future work, our simulations (\Cref{sec-experiments}) empirically reveal a clear trade-off: despite requiring deeper circuits per step, the $6$- and $8$-factor retractions converge in fewer iterations, ultimately reducing the total operational cost (measured by $e^{i\theta H}$ calls).
\end{remark}

\begin{landscape}
\begin{table}[!htbp]
\centering
\small
\renewcommand{\arraystretch}{2}

\caption{Summary of retractions on the unitary manifold.
Here \(\eta=\widetilde{\eta}U\) lies either in the full tangent space
\(T_U=\mathfrak{u}(N)U\) or in the fixed two-dimensional subspace
\(\mathcal{W}U\), where
\(\mathcal{W}:=\operatorname{span}_{\mathbb{R}}\{X_0,Y_0\}\subseteq\mathfrak{u}(N)\).
The matrices \(X_0\) and \(Y_0\) are determined by the underlying Grover search problem.
For Grover-compatible retractions, we write
\(\widetilde{\eta}=xX_0+yY_0\in\mathcal{W}\)
with unique coefficients \((x,y)\in\mathbb{R}^2\).
The column ``Preserve Grover plane'' indicates whether the updated state always remains in the Grover plane, while ``Grover-compatible'' indicates whether the retraction can be explicitly written as a finite product of Grover gates.
}
\label{tab:retractions-summary}

\begin{tabular}{@{}
L{0.12\linewidth}
L{0.45\linewidth}
C{0.10\linewidth}
C{0.15\linewidth}
C{0.10\linewidth}
@{}}
\hline
\textbf{Retraction}
& \textbf{Formula on unitary manifold}
& \textbf{Domain}
& \textbf{Preserve Grover plane}
& \textbf{Grover-compatible} \\
\hline

Riemannian exponential
&
\(
\Exp_U(\eta)=e^{\widetilde{\eta}}U
\)&
\(T_U\)
&
\ding{51}
&
\ding{55} \\

\hline

Cayley
&
\(
\mathrm{R}^{\mathrm{cay}}_U(\eta)
=
\left(I-\tfrac{1}{2}\widetilde{\eta}\right)^{-1}
\left(I+\tfrac{1}{2}\widetilde{\eta}\right)U
\)&
\(T_U\)
&
\ding{51}
&
\ding{55} \\

\hline

Polar
&
$\mathrm{R}^{\mathrm{polar}}_U(\eta) = (U+\eta)(I+\eta^\dagger\eta)^{-1/2}$
&
\(T_U\)
&
\ding{51}
&
\ding{55} \\

\hline

QR
&
$\mathrm{R}^{\mathrm{qr}}_U(\eta) = \operatorname{qf}(U+\eta)$, where \(\operatorname{qf}\) returns the \(Q\)-factor of a QR decomposition&
\(T_U\)
&
\ding{55}
&
\ding{55} \\

\hline

$8$-factor Grover-compatible
&
\(
\mathrm{R}^{(8)}_U(\eta)
=
e^{i\tfrac{y}{2}\psi_0}
e^{i\tfrac{\pi}{2}H}
e^{-i\tfrac{x}{2}\psi_0}
e^{-i\pi H}
e^{i\tfrac{x}{2}\psi_0}
e^{-i\tfrac{\pi}{2}H}
e^{-i\tfrac{y}{2}\psi_0}
e^{i\pi H}U
\)
&
\(\mathcal{W}U\)
&
\ding{51}
&
\ding{51} \\

\hline

$6$-factor Grover-compatible
&
$\mathrm{R}^{(6)}_U(\eta) = e^{i\tfrac{\pi}{2}H} e^{ic_1\psi_0} e^{-i\pi H} e^{ic_2\psi_0} e^{i\tfrac{\pi}{2}H} e^{iy\psi_0}U,$ where \( c_{1,2}=\mp\tfrac{x\pm y}{2}\)
&
\(\mathcal{W}U\)
&
\ding{51}
&
\ding{51} \\

\hline

$5$-factor Grover-compatible
&
\(\mathrm{R}^{(5)}_U(\eta) = e^{ia_1H} e^{ib_1\psi_0} e^{- i\pi H} e^{ib_2\psi_0} e^{-ia_2H}U,\) where \(A=\operatorname{atan2}(y,x)\), \(R=(x^2+y^2)^{1/2}\), \(a_{1,2}=A\pm\tfrac{\pi}{2}\), and \(b_{1,2}=\mp\tfrac{R}{2}\)&
\(\mathcal{W}U\)
&
\ding{51}
&
\ding{51} \\

\hline

General Grover-compatible
&
$\mathrm{R}_U(\eta)=(\prod_{\ell=1}^K e^{i \theta_{\ell}^{(1)}(x, y) H} e^{i \theta_{\ell}^{(2)}(x, y) \psi_0}) U$
&
\(\mathcal{W}U\)
&
\ding{51}
&
\ding{51} \\

\hline
\end{tabular}

\end{table}
\end{landscape}

\subsubsection{Grover-compatible Riemannian gradient ascent method}

Now, we provide the Riemannian gradient ascent (RGA) method employing Grover-compatible retractions in \cref{alg-grover-ret}. Moreover, the resulting state $\lvert \psi_T \rangle = U_T \lvert \psi_0 \rangle$ after $T$ iterations ($K T$ calls to $e^{i \theta H}$) is given by \cref{fig:grover_retr_update}.

\begin{algorithm}
\caption{Riemannian gradient ascent method using the Grover-compatible retraction (RGA-Grover)}
\label{alg-grover-ret}
\begin{algorithmic}[1]
\STATE{Choose a Grover-compatible retraction $\mathrm{R}$ (\cref{defn-grover-retraction}), explicitly given by $\mathrm{R}_{U}(t\eta) = V(t; x, y) U$, where  $V(t; x, y) = \prod_{\ell=1}^{K} e^{i \theta_\ell^{(1)}(t; x, y) H} e^{i \theta_\ell^{(2)}(t; x, y) \psi_0}$. Choose initial point $U_0=I$ and tolerance $\varepsilon$. Set $q_0 := \operatorname{Tr}(H\psi_0)$.}
\STATE{Set $k := 0$.}
\WHILE{$\left\|\operatorname{grad} f\left(U_k\right)\right\|> \varepsilon$}
  \STATE{Decompose $\widetilde{\operatorname{grad}} f (U_k)=[H,\psi_k] = x_k X_0 + y_k Y_0$ to find $(x_k,y_k) \in \R^2$.}
  \STATE{Choose a step size $t_k > 0$.}
  \STATE{Update
\begin{equation}\label{eq-grover-trt}
    U_{k+1}
= \mathrm{R}_{U_k}\bigl (t_k\, \operatorname{grad} f (U_k)\bigr)
    = V(t_k;x_k,y_k) U_k.
\end{equation}
\vspace{-5mm}
  }
  \STATE{Update the state $|\psi_{k+1}\rangle := V(t_k;x_k,y_k) |\psi_k \rangle$.}
  \STATE{Compute new cost value $q_{k+1} := \langle\psi_{k+1}|H |\psi_{k+1}\rangle$.}
  \STATE{Compute norm
  $\|\operatorname{grad} f(U_{k+1})\|_F = \sqrt{2q_{k+1}(1-q_{k+1})}$.}
  \STATE{Set $k := k+1$.}
\ENDWHILE
\end{algorithmic}
\end{algorithm}

\begin{figure}
    \centering
    \includegraphics[width=0.85\linewidth]{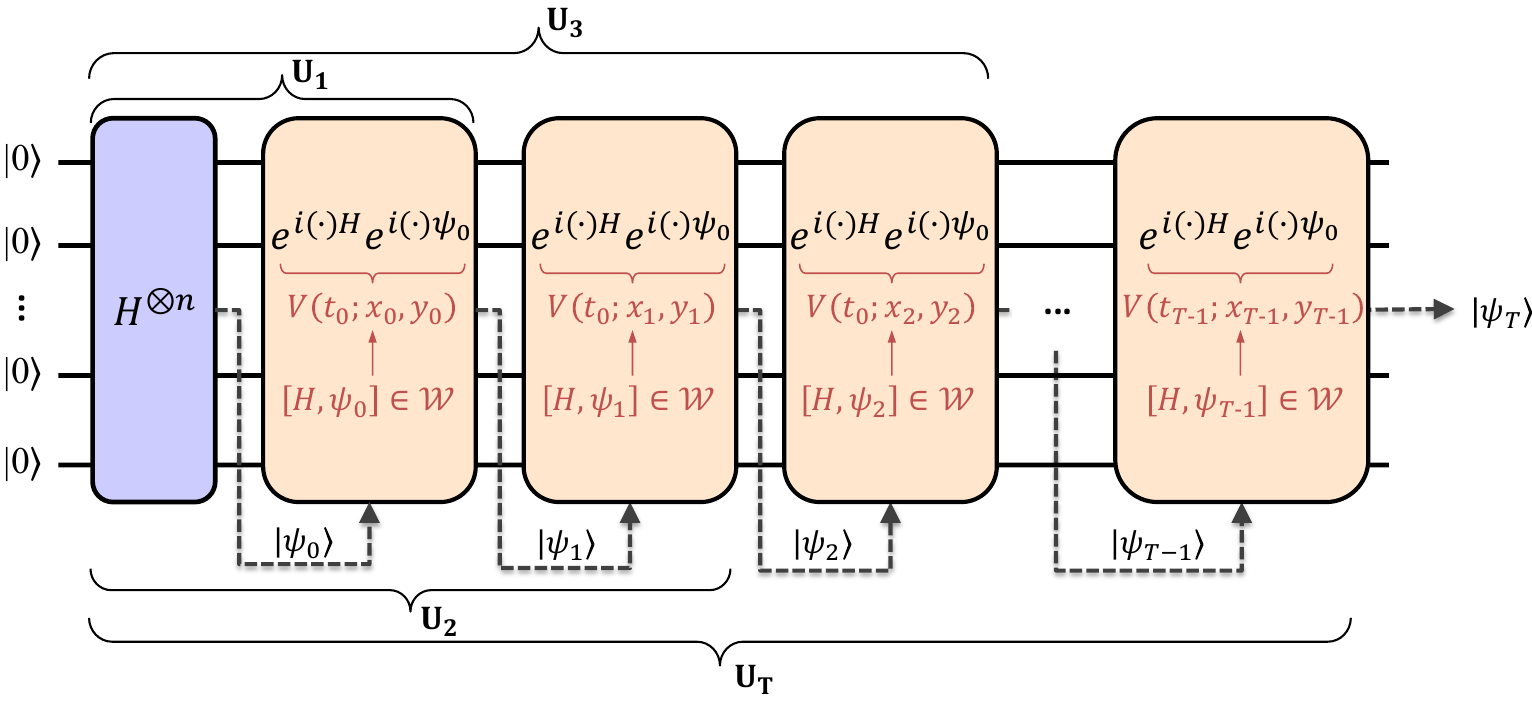}
    \caption{Quantum circuit generated by Riemannian gradient ascent method with the Grover-compatible retraction update \cref{eq-grover-trt}, starting from $U_0=I$ and the uniform state. At each iteration, a new gate $V(t_k;x_k,y_k)$ is appended.}
    \label{fig:grover_retr_update}
\end{figure}

This procedure is essentially identical to the standard Riemannian optimization algorithm, yet it preserves the defining feature of Grover's algorithm: every circuit consists solely of gates of the form $e^{i \theta H}$ and $e^{i \theta \psi_0}$. The key difference is that we now use gradient information to adaptively choose the rotation angle $\theta$.
At this point, Challenge 1 in \cref{sec-be-opt} has been resolved.

Challenge 2 can also be addressed by the following proposition.

\begin{proposition}\label{prop:challeng2}
If $|\psi_0\rangle$ is the uniform superposition over all items, then $|\psi_k\rangle$ in \cref{alg-grover-ret} converges to the uniform superposition over all marked items.
\end{proposition}

\begin{proof}
Let $q_0:=\langle\psi_0|H|\psi_0\rangle$ and define the target state $|\psi^\star\rangle:=(1/\sqrt{q_0}) H|\psi_0\rangle .$ By the proof of \cref{thm-grad-in-W}, every iterate satisfies $|\psi_k\rangle\in\mathcal S =\operatorname{span}_{\mathbb C}\{u,v\},$ where $u:=H|\psi_0\rangle$, $v:=(I-H)|\psi_0\rangle$ and $u\perp v$. Hence, for each $k$, there exist $\alpha_k,\beta_k\in\mathbb C$ such that $|\psi_k\rangle=\alpha_k u+\beta_k v.$ Recall that equation \cref{eq-1117} gives $q_k:=\langle\psi_k|H|\psi_k\rangle=|\alpha_k|^2q_0.$ On the other hand, we have \(\langle\psi^\star|\psi_k\rangle
=(1/\sqrt{q_0})\langle u|\alpha_k u+\beta_k v\rangle
=(1/\sqrt{q_0})\bigl(\alpha_k\langle u|u\rangle+\beta_k\langle u|v\rangle\bigr).\) Substituting $\langle u | u\rangle=q_0$ and $\langle u | v\rangle=0$, we obtain $\left|\left\langle \psi_k | \psi^\star \right\rangle\right|^2=|\alpha_k|^2 q_0.$ Thus,
\[
q_k
=
\left\langle\psi_k\right|H\left|\psi_k\right\rangle
=
\left|\left\langle\psi_k | \psi^\star\right\rangle\right|^2.
\]
Thus, when $q_k \to 1$, the states $|\psi_k\rangle$ and $|\psi^{\star}\rangle$ differ only by a global phase and are therefore physically indistinguishable.\footnote{For any unit complex vectors $|x\rangle$ and $|y\rangle$, the condition $\bigl|\langle x|y\rangle\bigr|=1$ holds if and only if the two vectors differ only by a global phase, i.e., $|y\rangle=e^{i\theta}|x\rangle$ for some $\theta\in\mathbb R$.}
In the Grover setting, with $H=H_g$ and $|\psi_0\rangle$ defined in \cref{eq-defn-Hf,eq-defn-psi0}, the state $|\psi^\star\rangle=(1/\sqrt{q_0})H|\psi_0\rangle$ coincides with the ideal target state in \cref{eq-defn-psistar}, namely the uniform superposition over all marked items.
\end{proof}

\subsection{Classical simulability of \texorpdfstring{\cref{alg-grover-ret}}{Algorithm 3.1}}

A main concern is how to compute the gradient coordinates $(x_k, y_k)$ in \cref{alg-grover-ret}. We know that $x_k$ (likewise for $y_k$) can be obtained via the inner product $x_k=\left\langle\left[H, \psi_k\right], X_0\right\rangle /(2 q_0(1-q_0))$, which seems to suggest that one might need to perform certain measurements on the current state $|\psi_k\rangle$ before proceeding with the circuit update. In fact, this is not an issue.

The following theorem shows that all key quantities generated throughout \cref{alg-grover-ret} at each iteration, including the cost values $q_k$, the gradient coordinates $(x_k, y_k)$ within $\mathcal{W}$, are classically simulatable. That is, the entire iterative process of \cref{alg-grover-ret} can be fully simulated on a classical computer. Once all angles have been computed, this Grover-type circuit can then be implemented accordingly.

\begin{theorem}[Classical simulability]\label{thm-simulability}
Consider the iterative process and notation introduced in \cref{alg-grover-ret}.
With initialization triplet $(x_0, y_0, q_0): = (1, 0, q_0)$, define the $2 \times 2$ matrix
$ \Psi_0 =
  \left(\begin{array}{ll}
q_0 & 1-q_0 \\
q_0 & 1-q_0
\end{array}\right).$
Then there exists an explicit, classically computable process
\begin{equation}\label{eq-F}
    F: (x_k, y_k, q_k; t_k) \mapsto (x_{k+1}, y_{k+1}, q_{k+1})
\end{equation}
described as follows. For each $k = 0, 1, \ldots,$
\begin{enumerate}
    \item If $k=0$, set $\alpha_k=1,$ $\beta_k=1$ and $z_k = \alpha_k \bar{\beta}_k$. ($\alpha_k,\beta_k,z_k \in \C$ are auxiliary variables)
    \item For the additional gate sequence in updates \cref{eq-grover-trt,eq-grover-form}, recall that $K$ denotes the number of oracle, diffusion gate pairs, while $\theta_{\ell}^{(1)}$ and $\theta_{\ell}^{(2)}$ denote the real-valued angle functions of the $\ell$-th pair. Thus,
\begin{equation*}
    V (t_k; x_k, y_k)
    = \prod_{\ell=1}^K
    e^{i \,\theta_{\ell}^{ (1)} (t_k; x_k, y_k) \,H}\,
    e^{i \,\theta_{\ell}^{ (2)} (t_k; x_k, y_k) \,\psi_0}.
\end{equation*}
We define the corresponding $2 \times 2$ matrix by
\begin{equation}\label{eq-Mk}
    M_k: = M (t_k; x_k, y_k)
    = \prod_{\ell=1}^{K}
    E_H\! \big (\theta_{\ell}^{ (1)} (t_k; x_k, y_k)\big)\,
    E_{\psi_0}\! \big (\theta_{\ell}^{ (2)} (t_k; x_k, y_k)\big),
\end{equation}
where $E_H(\theta)=\left(\begin{array}{cc}
e^{i \theta} & 0 \\
0 & 1
\end{array}\right)$ and $E_{\psi_0} (\theta) = I_2 + (e^{i\theta} - 1)\Psi_0$.
    \item Update
\begin{align*}
    & \left[\begin{array}{l}
    \alpha_{k+1} \\
    \beta_{k+1}
    \end{array}\right] \leftarrow M_k\left[\begin{array}{l}
    \alpha_k \\
    \beta_k
    \end{array}\right], \text{ and }
    \\
    & q_{k+1} \leftarrow q_0\left|\alpha_{k+1}\right|^2, \quad
    z_{k+1}\leftarrow \alpha_{k+1} \bar{\beta}_{k+1}, \quad
    x_{k+1} \leftarrow \Re z_{k+1}, \quad y_{k+1} \leftarrow \Im z_{k+1}.
\end{align*}
\end{enumerate}
\end{theorem}

\begin{proof}
We follow all notations introduced in the proof of \cref{thm-grad-in-W}.
Let $\mathcal{B}= \{u, v\}$ denote the orthogonal (but non-normalized) basis of the Grover plane $\mathcal{S}$, defined in \cref{eq-1125}.
By \cref{thm-grad-in-W}, for all $k$, the states $\lvert\psi_k\rangle$ lie in $\mathcal{S}$. Hence, we can write $\left|\psi_k\right\rangle=\alpha_k u+\beta_k v$ for some $\alpha_k, \beta_k \in \mathbb{C},$ or equivalently, write $\lvert\psi_k\rangle_{\mathcal{B}} = \binom{\alpha_k}{\beta_k} \in \mathbb{C}^2$.
We further define $z_k: = \alpha_k \bar{\beta}_k$ for all $k=0, 1, \ldots$, as in the proof of \cref{thm-grad-in-W}.
For $k=0$, since $\lvert\psi_0\rangle = u+v$, we have $\alpha_0 = 1$, $\beta_0 = 1$, and therefore $z_0 = 1$.
Since $\lvert\psi_{k+1}\rangle = V (t_k; x_k, y_k)\lvert\psi_k\rangle$, it follows that
\begin{equation*}
    \left|\psi_{k+1}\right\rangle_{\mathcal{B}}=M_k\left|\psi_k\right\rangle_{\mathcal{B}}, \quad \text { i.e., } \quad \left[\begin{array}{l}
\alpha_{k+1} \\
\beta_{k+1}
\end{array}\right] = M_k\left[\begin{array}{l}
\alpha_k \\
\beta_k
\end{array}\right],
\end{equation*}
where $M_k$ denotes the matrix representation of the operator $V (t_k; x_k, y_k)$. By \cref{eq-1117,eq-1123}, we have $q_{k+1} = q_0\lvert\alpha_{k+1}\rvert^2 $, $x_{k+1}=\Re\left (z_{k+1}\right),$ and $y_{k+1}=\Im\left (z_{k+1}\right).$ It remains to derive the explicit form of $M_k$.
We have already shown that
\begin{equation*}
[H]_{\mathcal{B}}=\left (\begin{array}{ll}
    1 & 0 \\
    0 & 0
    \end{array}\right), \quad \Psi_0: =\left[\psi_0\right]_{\mathcal{B}}=\left(\begin{array}{ll}
q_0 & 1-q_0 \\
q_0 & 1-q_0
\end{array}\right) .
\end{equation*}
Using the projector identity $e^{i\theta P} = I + (e^{i\theta} - 1)P$ for any $P^2 = P$, we define the following $2 \times 2$ matrices:
\begin{align*}
E_H (\theta)  =I_2+ (e^{i\theta}-1)[H]_{\mathcal{B}}=\left (\begin{array}{cc}
e^{i \theta} & 0 \\
0 & 1
\end{array}\right), \quad
E_{\psi_0} (\theta)  =I_2+ (e^{i\theta}-1)\Psi_0.
\end{align*}
Hence, the matrix representation of the operator $V (t_k; x_k, y_k)$ is given by \cref{eq-Mk}, which depends on the step size $t_k$ and the gradient coordinates $x_k, y_k$.
\end{proof}

\section{Complexity analysis}\label{sec-complexity}

It is well known that Grover's algorithm achieves a query complexity of $\mathcal{O}(\sqrt{N / M})$, yielding a quadratic speedup over classical search. Moreover, any quantum algorithm for the unstructured search problem must use at least $\Omega(\sqrt{N / M})$ queries \cite{zalka1999grover}, establishing that the $\mathcal{O}(\sqrt{N/M})$ bound is optimal.
Since the previous section showed how Grover's algorithm integrates cleanly into a Riemannian optimization framework, a natural question arises: \textit{Can standard complexity analysis techniques from optimization theory also recover the $\mathcal{O}(\sqrt{N/M})$ result?}

In this section, we recover the optimal scaling using tools from Riemannian optimization. We begin by carefully estimating the Riemannian Lipschitz constant $L_{\mathrm{Rie}}$ and show that $L_{\mathrm{Rie}} = \mathcal{O}(\sqrt{N/M})$. Next, a local Riemannian $\mu$-Polyak--\L{}ojasiewicz (PL) inequality with $\mu = \tfrac{1}{2}$ is established. Combining these two regularity conditions yields an iteration complexity of $\mathcal{O}(\sqrt{N/M} \log(1/\varepsilon))$ for \cref{alg-grover-ret} to reach an $\epsilon$-global maximizer, in agreement with Grover's quadratic speedup.

\subsection{Riemannian Lipschitz constants}

When optimizing a cost function in Euclidean space, its (gradient) Lipschitz constant is a crucial quantity in the analysis of iteration complexity.
Similarly, for optimizations defined on a manifold, one can also introduce notions of Riemannian Lipschitz continuity/constants intrinsic to the manifold.
In this subsection, we will compute the Riemannian Lipschitz constant of our cost function on the unitary manifold. We begin with the following lemma, which provides its Lipschitz constant in the Euclidean setting.

\begin{lemma}[Euclidean Lipschitz constant]\label{lem-lip-nablaf}
Consider the cost function $f (U)=\operatorname{Tr}\left (H U \psi_0 U^{\dagger}\right)$, then the Euclidean gradient $\nabla f(U)=2 H U \psi_0$ is Lipschitz continuous with (smallest) constant $L_{\mathrm{Euc}}= 2$, i.e.,
\begin{equation*}
    \left\|\nabla f\left (U_1\right)-\nabla f\left (U_2\right)\right\|_F \leq 2\left\|U_1-U_2\right\|_F,
\quad \text{ for all } U_1, U_2 \in \mathbb{C}^{N \times N}.
\end{equation*}
\end{lemma}
\begin{proof}
Since $H$ and $\psi_0$ are projectors, their spectral norms $\|H\|_2=\left\|\psi_0\right\|_2=1$. For any $U_1, U_2 \in \mathbb{C}^{N \times N}$,
\begin{align*}
2\left\|H\left (U_1-U_2\right) \psi_0\right\|_F
    \leq 2\|H\|_2\left\|\psi_0\right\|_2\left\|U_1-U_2\right\|_F
    =2\left\|U_1-U_2\right\|_F.
\end{align*}
Thus the smallest Lipschitz constant $L_{\mathrm{Euc}} \leq 2$.
We next show that it is tight. Let $w=\left|\psi_0\right\rangle$ so that $\psi_0=w w^{\dagger}$, and choose a unit vector $v \in \operatorname{range} (H)=\{v\in\mathbb C^N: Hv=v\}$ (since $H^2=H$). Set $A: =v w^{\dagger}$. Then $\|A\|_F=\|v\|\|w\|=1$,
and $\|H A \psi_0\|_F=\|H v w^{\dagger} w w^{\dagger}\|_F=\| (H v) w^{\dagger}\|_F=\|H v\| \|w\|=\|v\|=1.$ Choose $U_1, U_2$ such that $U_1-U_2=A$, then $\left\|\nabla f\left (U_1\right)-\nabla f\left (U_2\right)\right\|_F=2\left\|H A \psi_0\right\|_F=2\|A\|_F.$ Hence, $L_{\mathrm{Euc}}= 2.$
\end{proof}

Recall the notation $X_0 :=\left[H, \psi_0\right], $ $Y_0:=i\left[H, X_0\right],$ and $\mathcal{W} :=\operatorname{span}_{\mathbb{R}}\left\{X_0, Y_0\right\} \subseteq \mathfrak{u} (N).$ By \cref{lem-X0YO}, we define a constant
\begin{equation}\label{eq-c0}
c_0 :=\left\|X_0\right\|_F=\left\|Y_0\right\|_F=\left\|\left[X_0, \psi_0\right]\right\|_F=\frac{\sqrt{2 M (N-M)}}{N},
\end{equation}
since $c_0=\sqrt{2 q_0\left(1-q_0\right)}$ and $q_0=\operatorname{Tr}\left(H \psi_0\right)=\frac{M}{N}$.
Next, we establish two important bounds for Grover-compatible retractions. For concreteness, we present the proof using the 5-factor retraction as an example; the arguments for other retractions proceed in an analogous manner.

\begin{lemma}[Tight first and second order bounds]\label{lem-5factor-bounds}
Consider the 5-factor retraction $\mathrm{R}_U \equiv \mathrm{R}_U^{(5)}$ defined in \cref{eq-5factor-retraction}, and let $c_0$ be the constant defined in \cref{eq-c0}. For all $U \in \mathrm{U} (N)$ and $\eta \in \mathcal{W} U$, we have
\begin{equation*}
    \left\|\mathrm{R}_U (\eta)-U\right\|_F \leq\|\eta\|_F, \qquad\left\|\mathrm{R}_U (\eta)-U-\eta\right\|_F \leq \frac{1}{4 c_0}\|\eta\|_F^2.
\end{equation*}
Moreover, both are globally tight: $\tfrac{\left\|\mathrm{R}_U (\eta)-U\right\|_F}{\|\eta\|_F} \rightarrow 1,$ $\tfrac{\left\|\mathrm{R}_U (\eta)-U-\eta\right\|_F}{\|\eta\|_F^2} \rightarrow \frac{1}{4 c_0},$ as $\eta \rightarrow 0$.
\end{lemma}

\begin{proof}

Since $\eta \in \mathcal{W}U$, there exists a unique $\widetilde{\eta}\in\mathcal{W}$ such that $\eta=\widetilde{\eta}U$. We write its coordinates in the basis $\{X_0,Y_0\}$ as $\widetilde{\eta}=xX_0+yY_0,$  $x,y\in\mathbb{R},$ and set $R:=\sqrt{x^2+y^2}.$

Because $X_0$ and $Y_0$ are orthogonal under the Frobenius inner product and both have norm $\left\|X_0\right\|_F=\left\|Y_0\right\|_F=c_0$, we obtain
\begin{equation*}
    \|\eta\|_F=\|\widetilde{\eta}\|_F=c_0 \sqrt{x^2+y^2}=c_0 R .
\end{equation*}

We first establish the first-order bound. Consider the induced curve
\begin{equation}\label{eq-gamma-zero}
    \gamma (t): =\mathrm{R}_U (t \eta)=e^{i a_1 H} e^{i t b_1 \psi_0} \left (e^{-i \pi H} \right)e^{i t b_2 \psi_0} e^{-i a_2 H} U, \qquad t \in[0, 1].
\end{equation}
Notice that the middle term can be written as $e^{-i \pi H}=I+ (e^{-i \pi}-1)H=I-2 H$.
Differentiating $\gamma(t)$ with respect to $t$, only through the matrix exponentials involving $\psi_0$, gives
\begin{align}
    \gamma^{\prime} (t)&=e^{i a_1 H} e^{i t b_1\psi_0}\left (\frac{i R}{2}\left[e^{-i \pi H}, \psi_0\right]\right) e^{i t b_2 \psi_0} e^{-i a_2 H} U \notag\\
&=e^{i a_1 H} e^{i t b_1 \psi_0}\underbrace{\left (-i R X_0\right) }_{M: =}e^{i t b_2 \psi_0} e^{-i a_2 H} U, \label{eq-gamma-first}
\end{align}
where we used $\left[e^{-i \pi H}, \psi_0\right]=\left[I-2 H, \psi_0\right]=-2\left[H, \psi_0\right]=-2 X_0$.
Left and right multiplication by unitaries preserve the Frobenius norm, so $\left\|\gamma^{\prime} (t)\right\|_F=R\left\|X_0\right\|_F=\|\eta\|_F,$ $\forall t \in[0, 1].$ Using the integral form of the first-order Taylor expansion, $\gamma (1)=\gamma (0)+\int_0^1 \gamma^{\prime} (t) d t,$ we obtain
\begin{equation*}
    \left\|\mathrm{R}_U (\eta)-U\right\|_F=\left\|\int_0^1 \gamma^{\prime} (t) d t\right\|_F \leq \int_0^1\left\|\gamma^{\prime} (t)\right\|_F d t=\|\eta\|_F.
\end{equation*}
Because $\gamma^{\prime} (0)=\eta$, it follows that $\left\|\mathrm{R}_U (\eta)-U\right\|_F=\|\eta\|_F+o (\|\eta\|_F) \,(\eta \rightarrow 0).$ Thus, the constant 1 in the bound is globally tight.

We now establish the second-order bound. Differentiating $\gamma^{\prime} (t)$ in \cref{eq-gamma-first} once more, treating $M$ as $e^{-i\pi H}$ in \cref{eq-gamma-zero}, yields
\begin{align*}
    \gamma^{\prime \prime} (t) & =e^{i a_1 H} e^{i t b_1 \psi_0}\left (\frac{i R}{2}\left[M, \psi_0\right]\right) e^{i t b_2 \psi_0} e^{-i a_2 H} U \\
    & =e^{i a_1 H} e^{i t b_1 \psi_0}\left (\frac{R^2}{2}\left[X_0, \psi_0\right]\right) e^{i t b_2 \psi_0} e^{-i a_2 H} U.
\end{align*}
Therefore, $\left\|\gamma^{\prime \prime} (t)\right\|_F=\frac{R^2}{2}\left\|\left[X_0, \psi_0\right]\right\|_F=\frac{R^2}{2} c_0,$ $\forall t \in[0, 1].$ Using the integral form of the second-order Taylor expansion, $\gamma (1)=\gamma (0)+\gamma^{\prime} (0)+\int_0^1 (1-t) \gamma^{\prime \prime} (t) d t,$ we have
\begin{align*}
    \left\|\mathrm{R}_U (\eta)-U-\eta\right\|_F &=\left\| \int_0^1 (1-t) \gamma^{\prime \prime} (t) d t\right\|_F \\
    & \leq \int_0^1 (1-t)\left\|\gamma^{\prime \prime} (t)\right\|_F d t=\frac{1}{2} \cdot \frac{R^2}{2} c_0=\frac{1}{4 c_0}\|\eta\|_F^2.
\end{align*}
Moreover, since $\mathrm{R}_U (\eta)=U+\eta+\frac{1}{2} \gamma^{\prime \prime} (0)+o\left (\|\eta\|_F^2\right),$ and $\left\|\gamma^{\prime \prime} (0)\right\|_F=\frac{R^2}{2} c_0$,  it follows that $ \tfrac{\left\|\mathrm{R}_U (\eta)-U-\eta\right\|_F}{\|\eta\|_F^2}=\tfrac{1}{4 c_0}+o (1)\,(\eta \rightarrow 0),$ which establishes that the constant $1 / (4 c_0)$ is tight.
\end{proof}

We are now ready to compute the Riemannian Lipschitz constant of $f$.
Due to the intricacy of Riemannian geometry, extending Lipschitz continuity/constants to manifolds can be cumbersome, and several definitions appear in the literature. Here, we follow the approach recommended in \cite{boumal2019global}, which uses a pullback descent inequality. This method avoids additional geometric machinery and allows the analysis to be carried out entirely with Euclidean tools while still supporting the convergence results that follow.
The next \cref{prop-pullback-lip} establishes the Riemannian Lipschitz constant of the cost in \cref{pro-1} on the unitary manifold. Notably, such constant $L_{\mathrm{Rie}}$ scales as $\sqrt{N/M}$, increasing with the problem size $N = 2^n$.

\begin{proposition}[Riemannian Lipschitz constant]\label{prop-pullback-lip}
Consider the 5-factor retraction $\mathrm{R}_U \equiv \mathrm{R}_U^{(5)}$ defined in \cref{eq-5factor-retraction}. Let
\begin{equation}\label{eq-Lg}
L_{\mathrm{Rie}} :=2+\frac{N}{\sqrt{2 M (N-M)}} \in\mathcal{O} \left(\sqrt{\frac{N}{M}}\right).
\end{equation}
Then, for all $U \in \mathrm{U} (N)$, the pullbacks $f \circ \mathrm{R}_U: \mathcal{W}U \rightarrow \mathbb{R}$ satisfy
\begin{equation*}
\left|f\left (\mathrm{R}_U (\eta)\right)-\big[f (U)+\langle\operatorname{grad} f (U), \eta\rangle\big]\right| \leq \frac{L_{\mathrm{Rie}}}{2}\|\eta\|_F^2, \qquad \forall \eta \in \mathcal{W} U.
\end{equation*}
\end{proposition}
\begin{proof}
By \cref{lem-lip-nablaf}, the Euclidean gradient $\nabla f(U)$ is Lipschitz continuous over the entire Euclidean space $\mathbb{C}^{N \times N}$ with constant $L_{\mathrm{Euc}}= 2$. Hence, by the well-known descent lemma (see, e.g., \cite[Lemma 1.2.3]{nesterov_introductory_2004}), for all $U, W \in \mathrm{U} (N) \subseteq \mathbb{C}^{N \times N}$,
\begin{equation*}
    |f (W) -\left[f (U) + \langle \nabla f (U), W - U \rangle\right]| \leq \frac{L_{\mathrm{Euc}}}{2} \|W - U\|_F^2=\|W - U\|_F^2.
\end{equation*}
In particular, this inequality holds for $W := \mathrm{R}_U (\eta)$ with $\eta \in \mathcal{W}U$. Let $\Delta: = \mathrm{R}_U (\eta) - U \in \mathbb{C}^{N \times N}$, then
\begin{equation}\label{eq-1721}
    |f (\mathrm{R}_U (\eta)) -[f (U) + \langle \nabla f (U), \Delta \rangle]| \leq  \|\Delta\|_F^2 .
\end{equation}
Since the Riemannian gradient $\operatorname{grad} f (U)$ is the orthogonal projection of $\nabla f (U)$ onto $T_U$, see \cref{eq-gradf-formula}, we have
\begin{align}\label{eq-1457}
    \langle \nabla f (U), \Delta \rangle
    = \langle \nabla f (U), \eta +(\Delta - \eta)  \rangle
    = \langle \operatorname{grad} f (U), \eta \rangle + \langle \nabla f (U), \Delta - \eta \rangle .
\end{align}
Substituting \eqref{eq-1457} into \eqref{eq-1721} yields
\begin{equation*}
    |\underbrace{f (\mathrm{R}_U (\eta)) -[f (U) + \langle \operatorname{grad} f (U), \eta \rangle ]}_{A:=}-\underbrace{\langle \nabla f (U), \Delta - \eta \rangle}_{B:=} |
    \leq \|\Delta\|_F^2 .
\end{equation*}
Then, we have $|A|
    \leq |A - B| + |B|
    \leq \|\Delta\|_F^2 + |\langle \nabla f (U), \Delta - \eta \rangle| .$ By the Cauchy--Schwarz inequality and \cref{lem-5factor-bounds}, we obtain
\begin{align*}
    \big|A\big|
    & \leq \|\mathrm{R}_U (\eta) - U\|_F^2
    + \|\nabla f (U)\|_F  \|\mathrm{R}_U (\eta) - U - \eta\| _F \\
    & \leq \|\eta\|_F^2
    +  \frac{2}{4 c_0}\|\eta\|_F^2=(1+\frac{1}{2 c_0})\|\eta\|_F^2,
\end{align*}
where we used the fact $\|\nabla f (U)\|_F =\left\|2 H U \psi_0\right\|_F \leq 2\|H\|_2\|U\|_2\left\|\psi_0\right\|_F=2$. Setting $L_{\mathrm{Rie}}: =2 (1+\frac{1}{2 c_0}) =2+\frac{1}{c_0}$ and combining this with \cref{eq-c0} complete the proof.
\end{proof}

\subsection{Complexity results}

Regarding complexity, researchers in quantum computing and those in optimization focus on different aspects. Optimization researchers typically fix the problem size $N$ and study how the number of iterations $T$ depends on the target accuracy $\varepsilon$. In contrast, quantum computing researchers usually fix the accuracy $\varepsilon$ and investigate how the iteration count scales with the problem size $N$, namely, the number of qubits $n$. This point should be kept in mind when interpreting the complexity results below.

Using \cref{prop-pullback-lip} and extending the standard nonconvex complexity technique to the manifold setting \cite{boumal2019global}, we first obtain the following baseline complexity.

\begin{theorem}[Baseline Complexity]\label{thm-bese-complexity}
Consider \cref{pro-1} in the Grover setting, where the initial success probability is $q_0 = f (U_0) = M/N$. Suppose we run \cref{alg-grover-ret} with the 5-factor retraction $\mathrm{R}_U \equiv \mathrm{R}_U^{(5)}$ defined in \cref{eq-5factor-retraction}, and choose a fixed step size $t_k= 1 / L_{\mathrm{Rie}}$, where $L_{\mathrm{Rie}} = \mathcal{O}(\sqrt{N/M})$ is given in \cref{eq-Lg}.
Then, for any $\varepsilon>0$,
\begin{equation*}
    T \geq \left\lceil\frac{2L_{\mathrm{Rie}}}{\varepsilon^2} \right\rceil \quad \Longrightarrow \quad \min _{0 \leq k \leq T-1}\left\|\operatorname{grad} f\left (U_k\right)\right\| \leq \varepsilon.
\end{equation*}
\end{theorem}
\begin{proof}
This theorem is a direct application of \cite[Corollary 2.9]{boumal2019global}.
\end{proof}

For a fixed accuracy $\varepsilon$, the above result recovers the well-known $\mathcal{O}(\sqrt{N / M})$ complexity of Grover's algorithm by noting $L_{\rm Rie} = \mathcal{O}(\sqrt{N/M})$.
We will now show that the dependence on $\varepsilon$ can be improved from $1 / \varepsilon^2$ to $\log (1 / \varepsilon)$.
This refinement relies on establishing an explicit Riemannian PL condition for \cref{pro-1}, as shown below.

\begin{proposition}[PL condition of \cref{pro-1}]
The Riemannian PL inequality holds locally for problem \eqref{eq-cost}, namely, for any $U \in \mathrm{U} (N)$ such that $f(U) \geq \frac{1}{2}$,
\begin{equation}\label{eq:pl}
\|\grad f(U)\|^2 \geq 1 - f(U). \end{equation}
\end{proposition}
\begin{proof}
It follows from \cref{lem-X0YO} that $\|\operatorname{grad} f (U)\|^2
    = 2 f(U) (1-f(U)).$
Applying $f(U)\geq \frac{1}{2}$ yields the desired result.
\end{proof}

The $\mu$-PL inequality \cref{eq:pl}, first introduced in \cite{polyak1963gradient,lojasiewicz1963topological}, has become a central tool for establishing global linear convergence of $1 -\kappa^{-1}$ with $\kappa = L/\mu$ (where $L$ could be $L_{\rm Euc}$ or $L_{\rm Rie}$) for gradient-type methods \cite{karimi2016linear,rebjock2024fast} beyond the strongly convex setting. In the theorem below, we use it to obtain the global linear convergence of RGA.

\begin{theorem}[Best complexity]\label{thm-best-complexity}
Consider the same assumptions as in \cref{thm-bese-complexity}.
Let $q_k:=f\left(U_k\right) \in[0,1]$.
Then, for any $0<\varepsilon \leq q_0$, the iterates satisfy $1- q_T \leq \varepsilon$ in at most
\[
T=\left\lceil 6L_{\mathrm{Rie}} \log \left(\frac{1}{\varepsilon}\right) \right\rceil .
\]
\end{theorem}
\begin{proof}
Let $\delta_k: = 1 - q_k$. For each iterate $U_k$, by \cref{lem-X0YO} (or claim (2) of \cref{thm-grad-in-W}), we have $  \|\operatorname{grad} f (U_k) \|^2
    = 2 q_k (1 - q_k).$ By \cref{prop-pullback-lip}, for each $k$, the pullbacks $f \circ \mathrm{R}_{U_k}: \mathcal{W}U_k \rightarrow \mathbb{R}$ satisfy
\begin{align}
f\left (U_{k+1}\right)  = f\left (\mathrm{R}_{U_k} (\eta_k)\right) &\geq f (U_k)+\langle\operatorname{grad} f (U_k), \eta_k\rangle - \frac{L_{\mathrm{Rie}}}{2}\|\eta_k\|^2  \notag  \\
    &\geq f\left (U_k\right)+
(t_k-\frac{L_{\mathrm{Rie}}}{2} t_k^2)\left\|\operatorname{grad} f\left (U_k\right)\right\|^2 \notag  \\
& = f\left (U_k\right)+\frac{1}{2 L_{\mathrm{Rie}}}\left\|\operatorname{grad} f\left (U_k\right)\right\|^2,  \label{eq-retr-ascent}
\end{align}
where we used $\eta_k=t_k \operatorname{grad} f\left (U_k\right)$ together with the step size choice $t_k = 1/L_{\mathrm{Rie}}$ from \cref{alg-grover-ret}. In terms of $q_k = f (U_k)$, \cref{eq-retr-ascent} can be written as
\begin{equation}\label{eq-qk-0}
    q_{k+1}
     \geq
    q_k
    + \frac{1}{2 L_{\mathrm{Rie}}} \cdot 2 q_k (1 - q_k)
    = q_k + \frac{q_k (1 - q_k)}{L_{\mathrm{Rie}}}.
\end{equation}
In particular, the sequence $\{q_k\}_{k \geq 0}$ is nondecreasing.

\textbf{Phase I: From $q_0 = M/N$ to $q_k \geq \tfrac{1}{2}$.} Whenever $q_k \leq \tfrac{1}{2}$, we have $1 - q_k \geq \tfrac{1}{2}$, and thus \cref{eq-qk-0} implies
\begin{equation}\label{eq-half-1}
    q_{k+1}
     \geq
    q_k + \frac{q_k}{2 L_{\mathrm{Rie}}}
    = q_k \left (1 + \frac{1}{2 L_{\mathrm{Rie}}}\right).
\end{equation}
Let $K: = \min\left\{ k \geq 0: q_k \geq \tfrac{1}{2} \right\}$ be the first index at which the sequence reaches the level $\frac{1}{2}$. We construct an explicit upper bound for $K$. Define
\begin{equation*}
K_0
: = \left\lceil
    \frac{\log \left (\frac{1}{2 q_0}\right)}
    {\log \left (1 + \frac{1}{2 L_{\mathrm{Rie}}}\right)}
    \right\rceil \leq
\frac{\log \left (\frac{1}{2 q_0}\right)}
    {\log \left (1 + \frac{1}{2 L_{\mathrm{Rie}}}\right)}
+ 1.
\end{equation*}
We now verify that $K \leq K_0$. If some $k<K_0$ already satisfies $q_k \geq \frac{1}{2}$, then by definition $K \leq k<K_0$.
Otherwise, if $q_j<\frac{1}{2}$ holds for all $j<K_0$, then \cref{eq-half-1} applies to every such $j$, and induction yields
\begin{equation*}
    q_{K_0}
    \geq
    q_0 \left (1 + \frac{1}{2 L_{\mathrm{Rie}}}\right)^{K_0}
    \geq
    q_0 \cdot \frac{1}{2 q_0}
    = \frac{1}{2},
\end{equation*}
which again implies $K \leq K_0$.

Using the inequality $\log (1+x) \geq x / 2$ for $0<x \leq 1$, and substituting $x:=1 /\left(2 L_{\mathrm{Rie}}\right) \leq 1 / 4$, we obtain $\log \left(1 + \frac{1}{2 L_{\mathrm{Rie}}}\right)
\geq \frac{1}{4 L_{\mathrm{Rie}}}
\Rightarrow
\frac{1}{\log \left(1 + \frac{1}{2 L_{\mathrm{Rie}}}\right)}
\leq 4 L_{\mathrm{Rie}}.$ Hence,
\begin{equation}\label{eq-1006}
  K \leq K_0 \leq 4 L_{\mathrm{Rie}} \cdot \log \left (\frac{1}{2 q_0}\right) + 1.
\end{equation}
Therefore, the sequence $\left\{q_k\right\}$ reaches the level $q_k \geq \tfrac{1}{2}$ in at most $4 L_{\mathrm{Rie}} \log(1/(2q_0)) + 1$ iterations. This completes the analysis of Phase~I.

\textbf{Phase II: Linear convergence by the PL property in the region $q_k \geq \tfrac{1}{2}$.}
From Phase I, we have that for all $k \geq K$, $q_k = f(U_k) \geq \frac{1}{2}$.
Using \cref{eq-retr-ascent} and the PL inequality \cref{eq:pl}, it holds
\begin{equation*}
    \delta_{k+1} = 1 -f(U_{k+1}) \leq 1 - f(U_k) - \frac{1}{2L_{\rm Rie}} \| \grad f(U_k) \|^2
    \leq \left (1 - \frac{1}{2 L_{\mathrm{Rie}}}\right) \delta_k
\end{equation*}
for all $k \geq K$. By induction, this yields
\begin{equation*}
    \delta_k
    \leq \delta_K \left (1 - \frac{1}{2 L_{\mathrm{Rie}}}\right)^{k - K}\leq \left (1 - \frac{1}{2 L_{\mathrm{Rie}}}\right)^{k - K},
    \qquad  \forall k \geq K.
\end{equation*}
To hope $\delta_{T} \leq \varepsilon$ for some iteration $T \geq K$, it is sufficient to require $\left (1 - \frac{1}{2 L_{\mathrm{Rie}}}\right)^{T - K}
    \leq \varepsilon.$ Taking logarithms on both sides yields $(T - K)\log \left (1 - \frac{1}{2 L_{\mathrm{Rie}}}\right)
    \leq \log \left (\varepsilon\right),$ i.e.,
\begin{equation}\label{eq-T-K}
     T - K \geq \frac{\log \left (\varepsilon\right)}{\log \left (1 - \frac{1}{2 L_{\mathrm{Rie}}}\right)} = \frac{\log \left (\frac{1}{\varepsilon}\right)}{-\log \left (1 - \frac{1}{2 L_{\mathrm{Rie}}}\right)},
\end{equation}
since $\log \left (1 - \frac{1}{2L_{\mathrm{Rie}}}\right) < 0$. Since $\log (1-x) \leq-x$ for all $0<x<1$, we have $\frac{1}{-\log (1-x)} \leq \frac{1}{x}$.  Substituting $x:=1 /\left(2 L_{\mathrm{Rie}}\right) \leq 1 / 4$, we obtain $2 L_{\mathrm{Rie}}
    \geq \tfrac{1}{-\log \left (1 - \frac{1}{2 L_{\mathrm{Rie}}}\right)}.$ Thus a sufficient condition for \cref{eq-T-K} is
\begin{equation}\label{eq-1007}
    T - K
    \geq 2 L_{\mathrm{Rie}} \cdot \log\left(\frac{1}{\varepsilon}\right).
\end{equation}
In other words, if we choose $T \geq K + 2 L_{\mathrm{Rie}} \log\left(\frac{1}{\varepsilon}\right),$ then $\delta_{T} \leq \varepsilon.$

\textbf{Phase III: A unified complexity result.} Combining the bounds \cref{eq-1006,eq-1007} obtained in Phases~I and II, we obtain the unified sufficient condition
\begin{equation}\label{eq-4L2}
    T \geq 4 L_{\mathrm{Rie}} \log \left (\frac{1}{2 q_0}\right)+1+2 L_{\mathrm{Rie}} \log  \left (\frac{1}{\varepsilon}\right).
\end{equation}
We now show that the simpler requirement
\begin{equation}\label{eq-6Llog}
    T \geq 6 L_{\mathrm{Rie}} \log  \left (\frac{1}{\varepsilon}\right)
\end{equation}
implies \cref{eq-4L2} under the assumption $0<\varepsilon \leq q_0.$ Since $\frac{2 q_0}{\varepsilon} \geq 2$ and $L_{\mathrm{Rie}} \geq 2$, we have
\begin{align*}
    4 L_{\mathrm{Rie}} \log  \left (\frac{1}{\varepsilon}\right)-4 L_{\mathrm{Rie}} \log \left (\frac{1}{2 q_0}\right)
    =4 L_{\mathrm{Rie}} \log \left (\frac{2 q_0}{\varepsilon}\right) \geq 4 L_{\mathrm{Rie}} \log 2 \geq 8 \log 2>1 .
\end{align*}
Consequently, $4 L_{\mathrm{Rie}} \log \left (\frac{1}{\varepsilon}\right)+2 L_{\mathrm{Rie}} \log \left (\frac{1}{\varepsilon}\right) \geq \left (4 L_{\mathrm{Rie}} \log \frac{1}{2 q_0}+1\right)+2 L_{\mathrm{Rie}} \log \left (\frac{1}{\varepsilon}\right) .$
Therefore, whenever the condition \cref{eq-6Llog} holds, the more detailed sufficient condition \cref{eq-4L2} is automatically satisfied. Hence, under \cref{eq-6Llog}, we indeed have $\delta_{T} \leq \varepsilon.$
\end{proof}

It is worth noting that the obtained complexity $\mathcal{O}\bigl(\sqrt{N/M} \log(1/\varepsilon)\bigr)$ not only shows that the Grover-compatible Riemannian gradient ascent achieves a linear convergence rate for this nonconvex problem, but also matches the optimal $\sqrt{N/M}$ scaling of Grover-type algorithms.

We now clarify the distinction between two different types of conclusions. The invariance of the Grover plane $\mathcal{S}$ is a geometric property of the dynamics (\cref{thm-2d-dynamics,thm-grad-in-W}): once the current state lies in $\mathcal{S}$, the iterates generated by a Grover-compatible retraction remain in $\mathcal{S}$. This provides a two-dimensional description of the induced dynamics, in the same spirit as the exact analysis of original Grover's algorithm and its variants. In the original setting, this exact dynamics allows one to directly analyze the fixed angle choices $\alpha_k$ and $\beta_k$ (i.e., $\pi$), which yields the quadratic speedup. In contrast, the results in this section are optimization guarantees for the Riemannian gradient ascent method, where $\alpha_k$ and $\beta_k$ incorporate adaptive gradient information. These guarantees rely on the pullback Lipschitz constant and the PL inequality, and yield the same quadratic speedup from a different angle selection mechanism. Therefore, \cref{thm-bese-complexity,thm-best-complexity} should be understood as convergence guarantees from the optimization perspective, rather than as replacements for the exact global rotation analysis of original Grover's algorithm and its variants.

\section{Numerical simulations}\label{sec-experiments}

In this section, we validate the theoretical findings using numerical simulations.
Because \cref{alg-grover-ret} is classically simulatable due to \cref{thm-simulability}, all experiments are implemented directly using NumPy. All simulations were performed on a MacBook Pro (2021) with an Apple M1 Pro processor, 16 GB RAM. The source code is publicly available.\footnote{\texttt{\url{https://github.com/GALVINLAI/RGA_Grover}}}

Let $N=2^n$ for an $n$-qubit system, and unless otherwise noted we always take size of marked items $M=1$. The progress of the algorithm is measured by the cost value $q_k=f(U_k)=\operatorname{Tr}\left(H\psi_k \right)$ which corresponds to the success probability of identifying the marked item. Given an accuracy $\varepsilon>0$, the algorithm terminates once $\left|1-q_k\right|<\varepsilon$.

\paragraph{Representative optimization trajectory}
We begin by illustrating a typical optimization process. Set  $n=6$ and $\varepsilon=10^{-4}$, we apply the 5-factor Grover-compatible retraction with a fixed step size $t_k=1 / L_{\mathrm{Rie}}$; see \cref{eq-Lg}. \cref{fig:evolution_qk_fid_gradnorm} reports the evolution of the cost value $q_k$, and the norm of the Riemannian gradient  $\left\|\operatorname{grad}f(U_k)\right\|=\left\|[H,\psi_k]\right\|$. The curve of $q_k$ exhibits the smooth and monotone improvement characteristic of typical gradient ascent methods.

\begin{figure}
    \centering
    \includegraphics[width=0.6\linewidth]{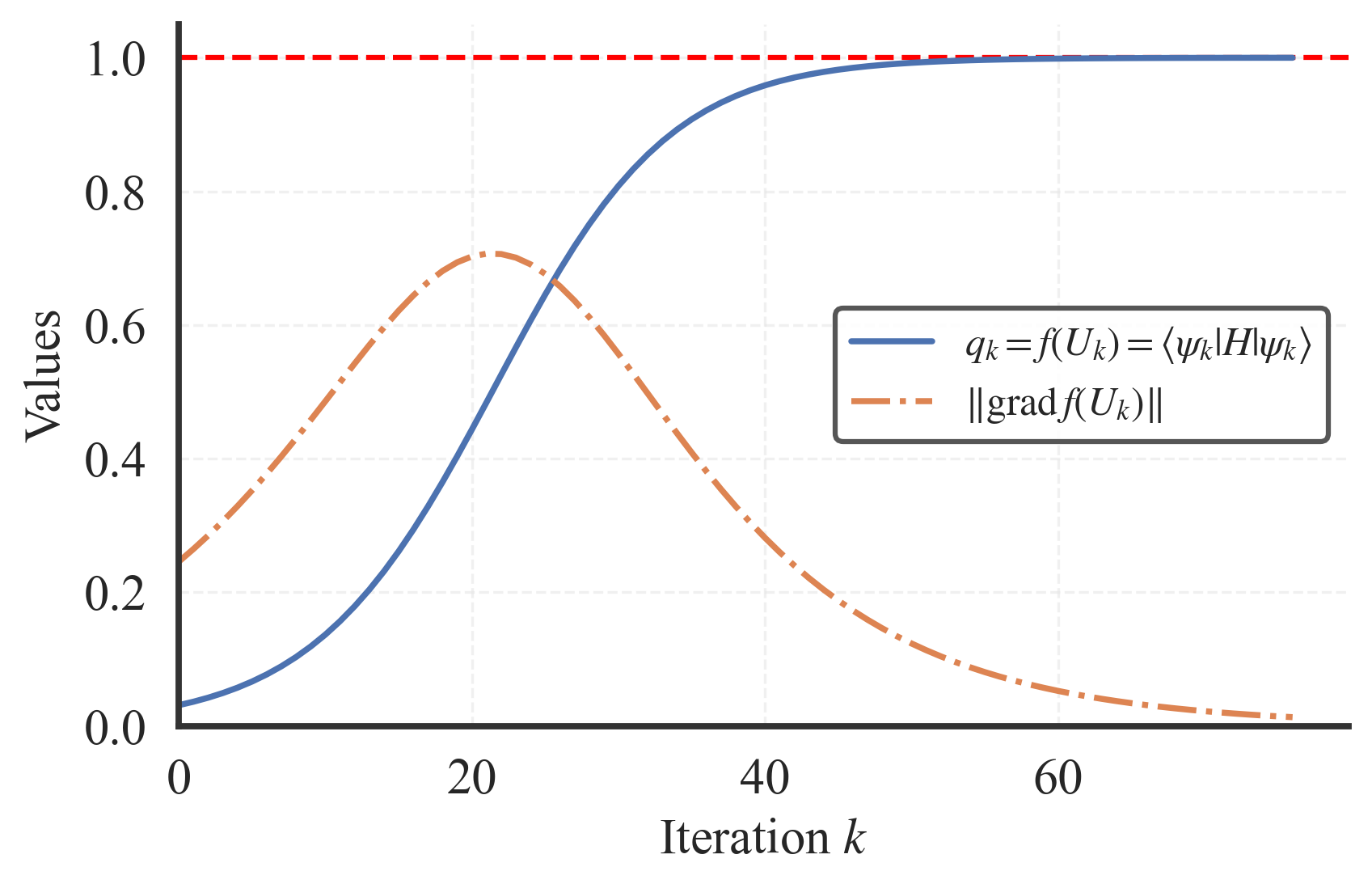}
    \caption{Representative optimization trajectory for $n = 6$ qubits using the 5-factor retraction with fixed step size $t_k = 1 / L_{\mathrm{Rie}}$. The cost value $q_k$ increases smoothly toward its maximal value $1$, illustrating the characteristic behavior of a typical gradient ascent method.}
    \label{fig:evolution_qk_fid_gradnorm}
\end{figure}

\paragraph{Scaling with problem size}
Next, we investigate how the iteration complexity scales with the problem size $N$.
We fix $\varepsilon=10^{-4}$ and vary the number of qubits from $n=2$ to $n=25$, so that $\sqrt{N}$ ranges from $2$ to $\sqrt{2^{25}} \approx 5792$.
In \cref{alg-grover-ret}, we consider two strategies for step sizes.
The first is the fixed $t_k = 1 / L_{\mathrm{Rie}}$, applied together with the 5-factor retraction.
The second is \textit{exact line search} (els), which requires solving a one-dimensional subproblem to determine the optimal step, i.e.,
\[
t_k := \arg\max_{t \geq 0}  f\bigl(\mathrm{R}_{U_k}(t\grad f(U_k))\bigr).
\]
Note that, thanks to the classical simulability stated in \cref{thm-simulability}, this line search computation can be carried out efficiently on a classical computer.
Recall the explicit process $F  (x_k, y_k, q_k; t_k) =(x_{k+1}, y_{k+1}, q_{k+1})$
described in \cref{eq-F}.
For fixed values of $(x_k, y_k, q_k)$, we employ a differential evolution solver to find $t_k$ that maximizes the resulting value of $q_{k+1}$, and then return the corresponding updated gradient coordinates $(x_{k+1}, y_{k+1})$.
The exact line search strategy does not require evaluating the Lipschitz constant $L_{\mathrm{Rie}}$. Thus, it can be conveniently applied to any of the three Grover-compatible retractions introduced in \cref{prop-grover-retr}: the 5-factor, 6-factor, and 8-factor retractions.

\cref{fig:iteration_comparison_total_H_calls} shows that, in all cases, the total number of $H$-exp (namely $e^{i \theta H}$) calls, equal to the iteration count $T$ multiplied by $2$, $3$, or $4$ respectively, exhibits a linear dependence on $\sqrt{N}$.
The slopes $s$ in the legends quantify this linear scaling.
\cref{fig:iteration_comparison_total_H_calls} (left) demonstrates that exact line search significantly improves efficiency, reducing the constant factor hidden in the complexity $T=\mathcal{O}(\sqrt{N})$, even though its theoretical analysis is more involved than that of the fixed $1 / L_{\mathrm{Rie}}$ strategy.
\cref{fig:iteration_comparison_total_H_calls} (right) zooms in on the three retractions under exact line search.
Interestingly, although the 6-factor and 8-factor retractions require more than two $H$-exp calls per iteration, they achieve fewer iterations overall, and thus yield a smaller total number of $H$-exp calls than the 5-factor retraction.

\begin{figure}
    \centering
    \includegraphics[width=0.85\linewidth]{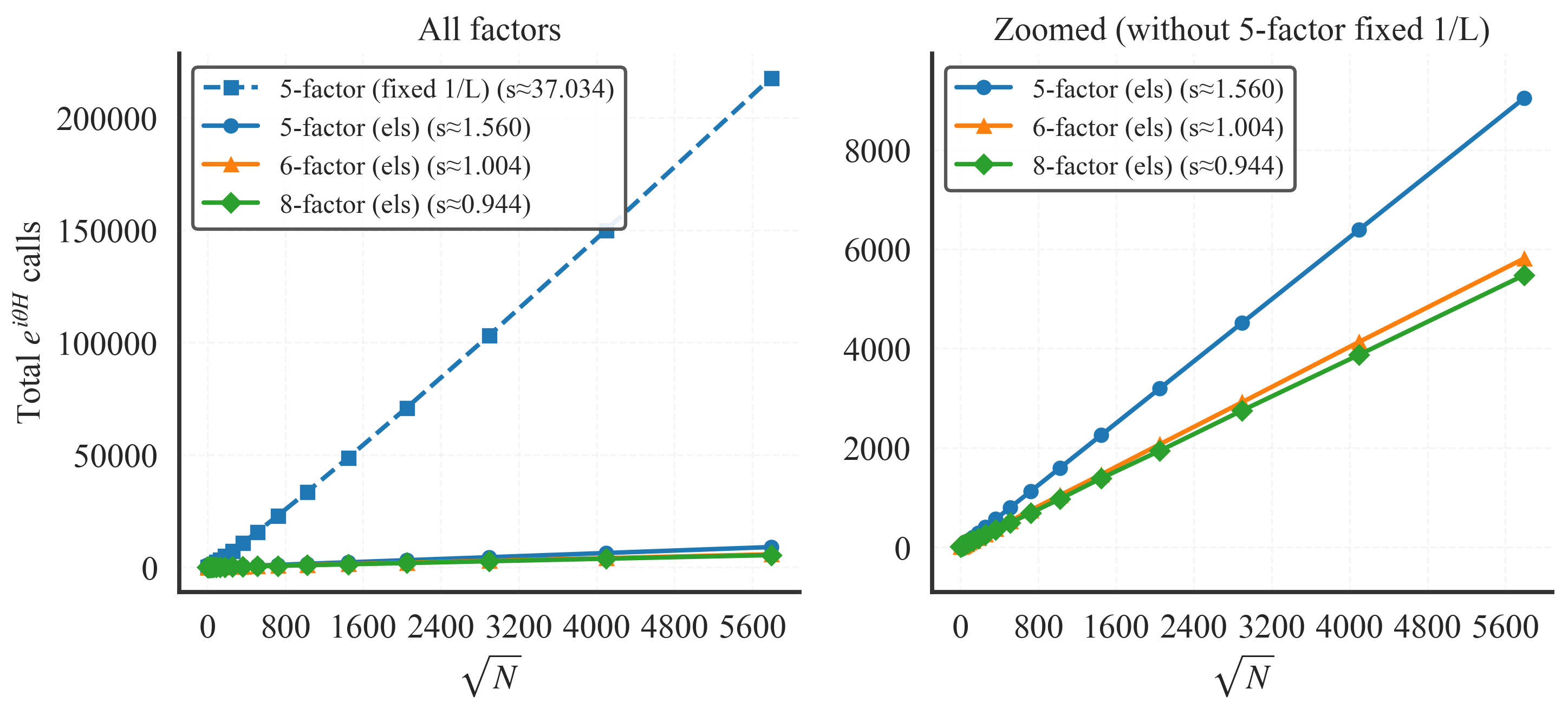}
    \caption{
Scaling of the total number of $H$-exp calls with problem size $\sqrt{N}$ for different Grover-compatible retractions and step size strategies.
Left: total $H$-exp calls for the 5-, 6-, and 8-factor retractions under fixed step size $t_k = 1 / L_{\mathrm{Rie}}$ and exact line search (els).
The fitted slopes $s \approx \text{const}$, shown in the legend, confirm the linear dependence $T = \mathcal{O}(\sqrt{N})$. Right: zoom on the exact line search, highlighting that the 6- and 8-factor retractions, despite requiring more $H$-exp evaluations per iteration, achieve fewer iterations overall and hence a smaller total number of $H$-exp calls than the 5-factor retraction.}
    \label{fig:iteration_comparison_total_H_calls}
\end{figure}

\paragraph{Scaling with accuracy}
Finally, we evaluate how the iteration complexity scales with the target accuracy $\varepsilon$.
We fix $n=15$ qubits and vary $\varepsilon$ from $10^{-2}$ to $10^{-12}$ in powers of 10.
We continue to use the same combination of retractions and step size strategies as in the previous experiments.
\cref{fig:H_calls_vs_tol_4methods} shows that, for all methods, the total number of applications of the $H$-exp operator grows linearly with $\log _{10}(1 / \varepsilon)$, in agreement with the theoretical bound. Over the range $\varepsilon \in\left[10^{-2}, 10^{-12}\right]$, the exact line search variants exhibit only a very mild increase (5-factor (els): $261 \rightarrow 290$ calls; 6-factor (els): $178 \rightarrow 210$; 8-factor (els): $162 \rightarrow 192$), whereas the fixed-step 5-factor retraction shows a much steeper growth (from 3942 to 9888 calls). This indicates that, even at high accuracy, employing exact line search can substantially improve efficiency.
Taken together, \cref{fig:iteration_comparison_total_H_calls,fig:H_calls_vs_tol_4methods} verify the complexity $\mathcal{O}(\sqrt{N} \log \left(1/\varepsilon\right))$ stated in \cref{thm-best-complexity}, and indicate that this bound is tight in practice.

\begin{figure}
    \centering
    \includegraphics[width=0.85\linewidth]{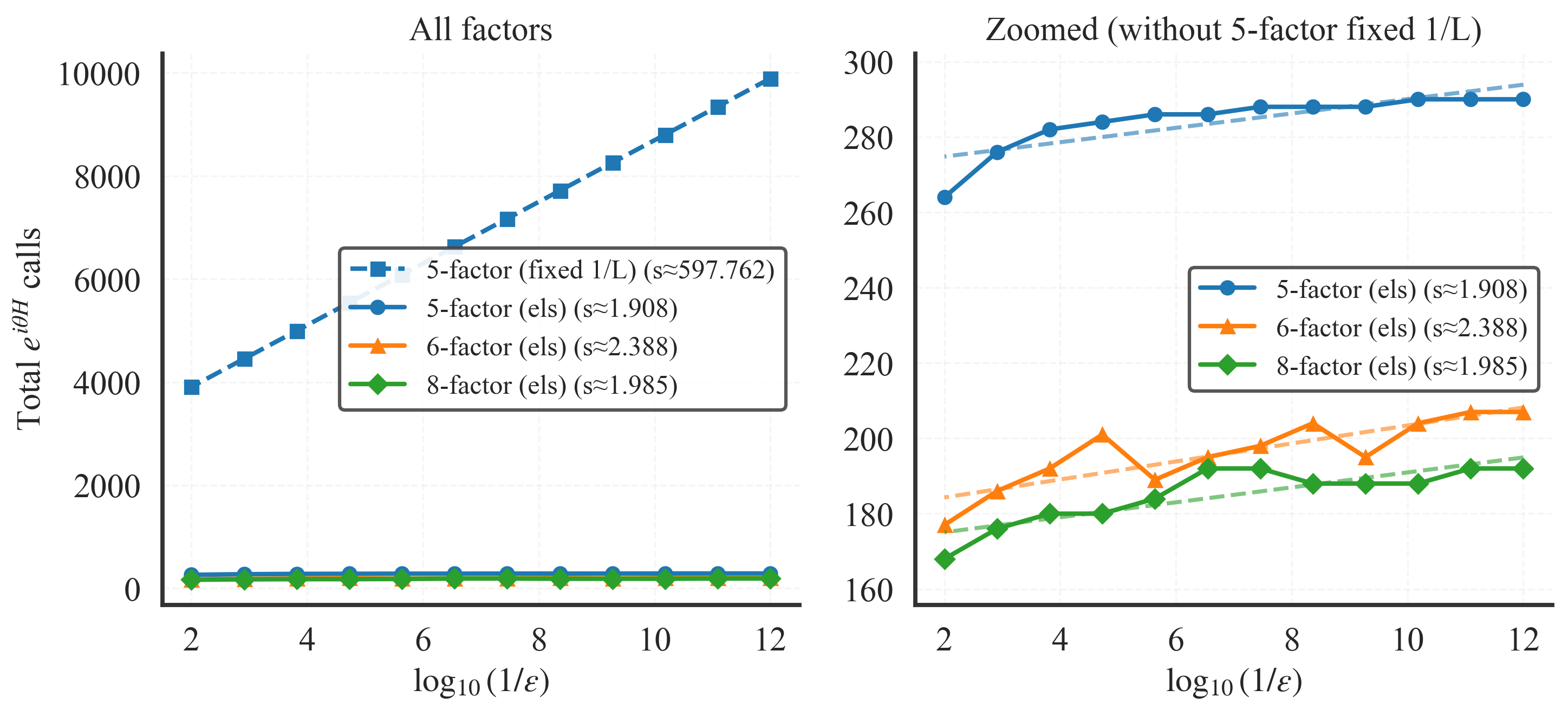}
    \caption{
Scaling of the total number of $H$-exp calls with the target accuracy $\varepsilon$ for different Grover-compatible retractions and step size strategies at fixed problem size $N = 2^{15}$.
Left: total $H$-exp calls for the 5-factor retraction with fixed step size $t_k = 1 / L_{\mathrm{Rie}}$ and for the 5-, 6-, and 8-factor retractions under exact line search (els), plotted against $\log_{10}(1/\varepsilon)$.
The fitted slopes $s \approx \text{const}$, shown in the legend, quantify the effective dependence on $\log_{10}(1/\varepsilon)$. Right: zoom on the three exact line search schemes, together with their fitted lines.
The iteration count $T$ almost grows linearly with $\log_{10}(1/\varepsilon)$.
}
\label{fig:H_calls_vs_tol_4methods}
\end{figure}

\section{Conclusions}\label{sec-discussion}

This work presents a new viewpoint on Grover's search algorithm through the lens of manifold optimization. Our key contribution is the Grover-compatible retraction, a physically implementable class of retractions. We further show that the optimization dynamics remain confined to a two-dimensional invariant subspace, rendering the algorithm classically simulatable. Using standard nonconvex optimization analysis, our approach recovers the optimal $\mathcal{O}(\sqrt{N})$ quadratic speedup, providing an optimization-theoretic explanation for this quantum advantage.
From an optimization standpoint, our Grover-compatible RGA method maps each mathematical iteration directly to a quantum circuit: each update appends gates generated by the current Riemannian gradient, with the rotation angle given by the step size. This provides a direct pathway for optimization theory to inform quantum circuit design on contemporary hardware.
We have several directions for future research:

\begin{itemize}

\item A systematic study of how different retractions affect performance. Our experiments show that the choice of retraction directly influences the pullback Lipschitz constant $L_{\text {Rie}}$, which determines the allowable step size and the resulting complexity bound. An optimal retraction would minimize this constant, yielding the sharpest guarantees and the fastest provable convergence.

\item Extending this framework to more advanced optimization methods. Although we focused on fixed-step size RGA (which depends on $L_{\text {Rie}}$), this viewpoint naturally invites the use of more sophisticated techniques, such as adaptive step sizes (e.g., Barzilai-Borwein \cite{fletcher2005barzilai,iannazzo2018riemannian}) and even second-order schemes such as Riemannian Newton methods.

\end{itemize}

\section*{Acknowledgment}
The authors used AI-based tools for polishing the English language.
Some proof strategies were interactively explored with large language models; all proofs were ultimately derived, checked, and written by the authors.
A detailed account of the human–AI collaboration is provided in the companion paper \cite{li2025groverAI}, with all prompts available at \texttt{\url{https://github.com/optsuite/MathResearchPrompts}}.

\bibliographystyle{siamplain}
\bibliography{references}

\end{document}


\maketitle

\section{Construction of Grover-compatible retractions}

This section gives a self-contained derivation of the Grover-compatible retractions used in the main text. The main idea is that constructing such a retraction is a first-order synthesis problem: given a tangent vector in a two-dimensional subspace, we construct a finite product of implementable Grover-type gates whose initial velocity equals that tangent vector.

\subsection{Basic setting}

Let \(H\in\mathbb C^{N\times N}\) and \(\psi_0\in\mathbb C^{N\times N}\) be Hermitian projectors, i.e., $H=H^\dagger = H^2,$ $\psi_0=\psi_0^\dagger = \psi_0^2.$ In the quantum search setting, \(H\) represents the marked-state projector and \(\psi_0=|\psi_0\rangle\langle\psi_0|\) is the rank-one projector onto the initial state. We work on the unitary manifold,
\[
    \mathrm{U} (N)=\bigl\{ U\in \C^{N\times N} \big| U^{\dagger}U=I_N\bigr\}.
\]
Its tangent space at \(U\in \mathrm U(N)\) is $T_U=\mathfrak u(N)U,$ where
\[
\mathfrak u(N)=\{\Omega\in\mathbb C^{N\times N}:\Omega^\dagger=-\Omega\}
\]
is the space of skew-Hermitian matrices. Define
\[
  X_0:=[H,\psi_0],\qquad
  Y_0:=i[H,X_0],\qquad
  \mathcal W:=\operatorname{span}_{\mathbb R}\{X_0,Y_0\}
  \subseteq \mathfrak u(N).
\]
Both \(X_0\) and \(Y_0\) are skew-Hermitian, so \(\mathcal W\) is a real two-dimensional subspace of \(\mathfrak u(N)\). The restricted tangent subspace of interest at \(U\) is \(\mathcal WU\subseteq T_U\). For any \(\eta\in\mathcal WU\), there are unique real coefficients \(x,y\) such that the skew-Hermitian component of $\eta$ satisfies
\[
  \widetilde\eta:=\eta U^\dagger=xX_0+yY_0\in\mathcal W.
\]
Thus, we identify \(\mathcal W\cong\mathcal WU\cong\mathbb R^2\) via the coordinate pair \((x,y)\).

We now recall the definition of Grover-compatible retractions as in the main text (see \cref{defn-grover-retraction}). Let \(U\in \mathrm U(N)\). A mapping
\[
 \mathrm{R}_U:\mathcal WU\rightarrow \mathrm U(N)
\]
is called a Grover-compatible retraction if it satisfies the following two
conditions.

\begin{enumerate}
    \item First, \(\mathrm{R}_U\) is a retraction on the restricted tangent subspace
\(\mathcal WU\). Namely, for every \(\eta\in \mathcal WU\), the curve
\(\gamma(t):=\mathrm{R}_U(t\eta)\), \(t\geq 0\), satisfies
\[
  \gamma(0)=U,\qquad \dot\gamma(0)=\eta.
\]

\item Second, the retraction is implementable by Grover-type gates. More precisely,
if \(\eta U^\dagger=xX_0+yY_0\), then
\[
 \mathrm{R}_U(t\eta)=V(t;x,y)U,
\]
where \(V(t;x,y)\) is a finite product of gates generated only by \(H\) and \(\psi_0\):
\begin{equation}\label{eq-Vtxy}
    V(t;x,y)
    =
    \prod_{\ell=1}^K
    e^{i\theta_\ell^{(1)}(t;x,y)H}
    e^{i\theta_\ell^{(2)}(t;x,y)\psi_0},
\end{equation}
for some finite integer \(K\) and real-valued parameter functions \(\theta_\ell^{(1)}\) and \(\theta_\ell^{(2)}\).
\end{enumerate}

The above definition reduces the construction of Grover-compatible retractions to a first-order synthesis problem. Indeed, \(\gamma(0)=U\) is equivalent to \(V(0;x,y)=I\). Moreover,
\[
  \left.\frac{d}{dt}\mathrm{R}_U(t\eta)\right|_{t=0}
  =
  \dot V(0;x,y)U.
\]
Hence the initial velocity condition \( \dot\gamma(0)=\eta\) is equivalent to $\dot V(0;x,y)=xX_0+yY_0.$ Therefore, it suffices to construct a finite product \(V(t;x,y)\) of \(H\)- and \(\psi_0\)-generated gates satisfying
\[
  V(0;x,y)=I,\qquad \dot V(0;x,y)=xX_0+yY_0.
\]
This is the first-order synthesis solved below.

\subsection{The elementary building block}

For fixed real numbers \(a\) and \(b\), consider the following one-parameter curve on the unitary manifold:
\begin{equation}\label{23213212}
  C_{a,b}(t):=e^{iaH}e^{itb\psi_0}e^{-iaH}.
\end{equation}
Here \(t\) is the curve parameter. In this form, the two \(H\)-generated conjugation gates are independent of \(t\); the parameter \(t\) appears only in the middle \(\psi_0\)-generated gate. Later, we will use this ``building block'' as the basic component for constructing the 5-factor, 6-factor, and 8-factor Grover-compatible retractions in the main text.

The usefulness of this block comes from its first-order behavior at \(t=0\). It starts from the identity,
\[
  C_{a,b}(0)=I,
\]
and its initial velocity is
\begin{equation}\label{111dfs}
  \left.\frac{d}{dt}C_{a,b}(t)\right|_{t=0}
  =
  b\,i e^{iaH}\psi_0e^{-iaH}
  =
  b\bigl(i\psi_0-\sin(a)X_0+(\cos(a)-1)Y_0\bigr).
\end{equation}
A proof of this identity is given below. This formula reveals the basic construction strategy. A single block \(C_{a,b}(t)\) produces first-order components in the desired \(X_0\) and \(Y_0\) directions, but it also produces an unwanted \(i\psi_0\) component. Therefore, we will combine several such blocks so that the \(i\psi_0\) components cancel, while the remaining \(X_0\) and \(Y_0\) components match the prescribed coefficients \(x\) and \(y\). In the resulting constructions, the parameters \(a\) and \(b\) will be
specified explicitly in terms of the target coordinates \((x,y)\).

\begin{proof}[Proof of Initial Velocity Equation \cref{111dfs}]
Since \(a\) and \(b\) are fixed parameters, only the middle factor depends on
\(t\). Therefore,
\begin{align*}
  \left.\frac{d}{dt}C_{a,b}(t)\right|_{t=0}
  =
  e^{iaH}\left(
  \left.\frac{d}{dt}e^{itb\psi_0}\right|_{t=0}\right)
  e^{-iaH}
  =
  e^{iaH}(ib\psi_0)e^{-iaH}
  =
  b\,i e^{iaH}\psi_0e^{-iaH}.
\end{align*}
It remains to expand the conjugating factors $e^{\pm iaH}$. Since \(H^2=H\), then $ e^{\pm iaH}=I+(e^{\pm ia}-1)H.$ Thus
\begin{align*}
  i e^{iaH}\psi_0e^{-iaH}
  &=
  i\Bigl(I+(e^{ia}-1)H\Bigr)
  \psi_0
  \Bigl(I+(e^{-ia}-1)H\Bigr) \\
  &=
  i\psi_0
  +i(e^{ia}-1)H\psi_0
  +i(e^{-ia}-1)\psi_0H
  +i(e^{ia}-1)(e^{-ia}-1)H\psi_0H.
\end{align*}
Using $e^{\pm ia}-1=(\cos a-1)\pm i\sin a,$ and $(e^{ia}-1)(e^{-ia}-1)=-2(\cos a-1),$ we obtain
\begin{align*}
  i e^{iaH}\psi_0e^{-iaH}
  &=
  i\psi_0
  +\bigl(i(\cos a-1)-\sin a\bigr)H\psi_0 \\
  &\quad
  +\bigl(i(\cos a-1)+\sin a\bigr)\psi_0H
  -2i(\cos a-1)H\psi_0H \\
  &=
  i\psi_0
  -\sin a\,(H\psi_0-\psi_0H) \\
  &\quad
  +(\cos a-1)i
  \bigl(H\psi_0+\psi_0H-2H\psi_0H\bigr).
\end{align*}
Finally, by the definitions
\[
  X_0:=H\psi_0-\psi_0H,\qquad
  Y_0:=i[H,X_0]=i(H\psi_0+\psi_0H-2H\psi_0H),
\]
we get
\[
  i e^{iaH}\psi_0e^{-iaH}
  =
  i\psi_0-\sin(a)X_0+(\cos(a)-1)Y_0.
\]
Substituting this into the derivative formula proves
\[
  \left.\frac{d}{dt}C_{a,b}(t)\right|_{t=0}
  =
  b\bigl(i\psi_0-\sin(a)X_0+(\cos(a)-1)Y_0\bigr).
\]
\end{proof}

\begin{remark}
The particular choice of \(C_{a,b}(t)\) in \cref{23213212} is motivated by its simple structure and by the fact that, as shown below, it is already flexible enough to generate several useful Grover-compatible retractions. We do not claim that this is the only possible elementary building block. For example, one may also interchange the roles of \(H\) and \(\psi_0\) and consider the alternative curve
\[
  \widetilde C_{a,b}(t)
  :=
  e^{ia\psi_0}e^{itbH}e^{-ia\psi_0}.
\]
Such variants may lead to different but analogous first-order synthesis schemes.
\end{remark}

\begin{remark}
The key initial velocity identity in \cref{111dfs} is also the underlying ingredient in the proof of \cref{prop-grover-retr} in the main text. Due to space limitations, the main text does not expand this calculation in detail. Instead, \cref{prop-grover-retr} directly states several concrete Grover-compatible retractions and verifies them by checking the retraction conditions. In contrast, this section derives these formulas constructively: starting from the definition, we use the elementary block \(C_{a,b}(t)\) and the coefficient matching principle to explain how the 5-factor, 6-factor, and 8-factor formulas arise.
\end{remark}

\subsection{General coefficient matching principle}

Consider a product of \(J\) elementary blocks,
\[
  V_J(t):=\prod_{j=1}^J C_{a_j,b_j}(t).
\]
In the following constructions, we build the product \(V(t;x,y)\) in \cref{eq-Vtxy} precisely by choosing suitable \(a_j\) and \(b_j\) in such a product. Since \(C_{a_j,b_j}(0)=I\) for each \(j\), we have
\[
  V_J(0)=I.
\]
Moreover, the product rule, with differentiation taken with respect to \(t\), gives
\[
  \dot V_J(0)
  =  \sum_{j=1}^J\left.\frac{d}{d t} C_{a_j,b_j}(t)\right|_{t=0}
  =
  \sum_{j=1}^J
  b_j\bigl(i\psi_0-\sin(a_j)X_0+(\cos(a_j)-1)Y_0\bigr).
\]
Thus, the condition \(\dot V_J(0)=xX_0+yY_0\) if and only if
\[
\left\{
\begin{aligned}
  \sum_{j=1}^J b_j &=0, \\
  -\sum_{j=1}^J b_j\sin(a_j)&=x, \\
  \sum_{j=1}^J b_j(\cos(a_j)-1)&=y.
\end{aligned}
\right.
\]
Because the first equation gives \(\sum_j b_j=0\), the last equation becomes \(\sum_j b_j\cos(a_j)=y\). Hence the conditions may also be written compactly as \textit{matching equations}:
\[
\left\{
\begin{aligned}
  \sum_{j=1}^J b_j &= 0, \\
  -\sum_{j=1}^J b_j\sin(a_j) &= x, \\
  \sum_{j=1}^J b_j\cos(a_j) &= y.
\end{aligned}
\right.
\]
Every construction below is obtained by choosing the number of blocks \(J\) and then solving the coefficient matching equations for the parameters \(a_j\) and \(b_j\) in terms of the target coordinates \((x,y)\).

\subsection{One building block: failure}

For \(J=1\), the matching equations are
\[
  b_1=0,\qquad -b_1\sin(a_1)=x,\qquad b_1 \cos(a_1)=y.
\]
Therefore \(x=y=0\). A single building block can only produce the zero tangent vector. This failure is unavoidable: the undesired \(i\psi_0\) component in \cref{111dfs} from one block cannot be cancelled by any other term. Thus, any nontrivial Grover-compatible retraction constructed from
these elementary blocks requires at least two blocks.

\subsection{Two building blocks: a 5-factor family}

For \(J=2\), we take
\[
  V_2(t)=C_{a_1,b_1}(t)C_{a_2,b_2}(t).
\]
The matching equations are
\[
\left\{
\begin{aligned}
b_1+b_2&=0, \\
  -b_1\sin a_1-b_2\sin a_2&=x, \\
  b_1 \cos a_1 +b_2 \cos a_2 &=y.
\end{aligned}
\right.
\]
The first equation cancels the unwanted \(i\psi_0\) component. We set $b_1=-\beta,$  $b_2=\beta.$ Then the remaining two equations reduce to
\[
  \beta(\sin a_1-\sin a_2)=x,\qquad
  \beta(\cos a_2-\cos a_1)=y.
\]
Introduce the midpoint and half difference of the two angles:
\[
  \mu:=\frac{a_1+a_2}{2},\qquad
  \delta:=\frac{a_1-a_2}{2}.
\]
Equivalently,
\[
  a_1=\mu+\delta,\qquad a_2=\mu-\delta.
\]
Using the trigonometric identities, we obtain
\begin{gather*}
  x
=\beta(\sin a_1-\sin a_2)
  =\beta\bigl(\sin(\mu+\delta)-\sin(\mu-\delta)\bigr)
  =2\beta\sin\delta\cos\mu,\\
  y
  =\beta(\cos a_2-\cos a_1)
  =\beta\bigl(\cos(\mu-\delta)-\cos(\mu+\delta)\bigr)
  =2\beta\sin\delta\sin\mu.
\end{gather*}
Hence the two-block construction produces
\[
  (x,y)=2\beta\sin\delta\,(\cos\mu,\sin\mu).
\]
On the other hand, we may write $(x,y)$ by polar coordinates:
\[
  (x,y)=R(\cos A,\sin A),
  \qquad
  R:=\sqrt{x^2+y^2},\qquad
  A:=\operatorname{atan2}(y,x).
\]
Thus, one should force
\[
  2\beta\sin\delta=R,\qquad
  \mu=A.
\]
Therefore, for any \(\delta\) satisfying \(\sin\delta\neq 0\), we may choose
\[
  a_1=A+\delta,\qquad
  a_2=A-\delta,\qquad
  b_1=-\frac{R}{2\sin\delta},\qquad
  b_2=\frac{R}{2\sin\delta}.
\]
Finally, \(\delta\) is a free parameter of the two-block
construction, subject only to \(\sin\delta\neq 0\).

The 5-factor retraction in the main text is obtained by the convenient choice
\(\delta=\pi/2\). Then
\[
  a_1=A+\frac{\pi}{2},\qquad
  a_2=A-\frac{\pi}{2},\qquad
  b_1=-\frac R2,\qquad
  b_2=\frac R2.
\]
Moreover \(a_2-a_1=-\pi\), so the middle \(H\)-generated factor becomes \(e^{-i\pi H}\). This gives
\[
  V_2(t)
  =
  e^{ia_1H}e^{itb_1\psi_0}
  e^{-i\pi H}
  e^{itb_2\psi_0}e^{-ia_2H}.
\]
Setting \(t=1\) gives
\[
 \mathrm{R}_U^{(5)}(\eta)
  =
  e^{ia_1H}e^{ib_1\psi_0}
  e^{-i\pi H}
  e^{ib_2\psi_0}e^{-ia_2H}U.
\]
This derivation shows that the 5-factor formula is not unique. Rather, it is a one-parameter family.

\subsection{Three building blocks: a 7-factor family}

For \(J=3\), we consider
\[
  V_3(t)=C_{a_1,b_1}(t)C_{a_2,b_2}(t)C_{a_3,b_3}(t).
\]
In this case, the matching equations can be written as a linear system in \(b_1,b_2,b_3\):
\[
  \begin{bmatrix}
  1 & 1 & 1\\
  -\sin a_1 & -\sin a_2 & -\sin a_3\\
  \cos a_1 & \cos a_2 & \cos a_3
  \end{bmatrix}
  \begin{bmatrix}
  b_1\\ b_2\\ b_3
  \end{bmatrix}
  =
  \begin{bmatrix}
  0\\ x\\ y
  \end{bmatrix}.
\]
We denote this coefficient matrix by \(M(a_1,a_2,a_3)\). The solutions
\(b_1,b_2,b_3\) are uniquely determined whenever \(M(a_1,a_2,a_3)\) is
nonsingular.

Indeed, the nonsingularity condition admits a simple algebraic characterization.
Expanding the determinant along the first row and using trigonometric
identities, we obtain
\[
  \det M(a_1,a_2,a_3)
  =
  4
  \sin\frac{a_1-a_2}{2}
  \sin\frac{a_2-a_3}{2}
  \sin\frac{a_3-a_1}{2}.
\]
Therefore, $M\left(a_1, a_2, a_3\right)$ is nonsingular if and only if the three angles \(a_1,a_2,a_3\) are pairwise distinct modulo
\(2\pi\).

This gives a general construction method.
Fix any three angles \(a_1,a_2,a_3\) that are pairwise distinct
modulo \(2\pi\). These three angles serve as free parameters of the construction. Define
\[
   \begin{bmatrix}
  b_1(x,y)\\ b_2(x,y)\\ b_3(x,y)
  \end{bmatrix}\equiv
 \begin{bmatrix}
  b_1\\ b_2\\ b_3
  \end{bmatrix}
  :=
  M(a_1,a_2,a_3)^{-1}
  \begin{bmatrix}
  0\\ x\\ y
  \end{bmatrix}.
\]
Then the resulting \(V_3(t)\) satisfies \(\dot V_3(0)=xX_0+yY_0\).
Expanding \(V_3(t)\) and merging adjacent \(H\)-generated factors gives
\[
  V_3(t)
  =
  e^{ia_1H}e^{itb_1\psi_0}
  e^{i(a_2-a_1)H}e^{itb_2\psi_0}
  e^{i(a_3-a_2)H}e^{itb_3\psi_0}
  e^{-ia_3H}.
\]
Thus, setting \(t=1\), we obtain the general \(7\)-factor retraction
\[
 \mathrm{R}_U^{(7)}(\eta)
  =
  e^{ia_1H}e^{ib_1\psi_0}
  e^{i(a_2-a_1)H}e^{ib_2\psi_0}
  e^{i(a_3-a_2)H}e^{ib_3\psi_0}
  e^{-ia_3H}U,
\]
where \(b_1,b_2,b_3\) are determined by the above linear system.

The \(6\)-factor retraction introduced in the main text is obtained as a special case by choosing
\[
  a_1=\frac{\pi}{2},\qquad
  a_2=-\frac{\pi}{2},\qquad
  a_3=0.
\]
For this choice, the matching system becomes
\[
  \begin{bmatrix}
  1 & 1 & 1\\
  -1 & 1 & 0\\
  0 & 0 & 1
  \end{bmatrix}
  \begin{bmatrix}
  b_1\\ b_2\\ b_3
  \end{bmatrix}
  =
  \begin{bmatrix}
  0\\ x\\ y
  \end{bmatrix}.
\]
Solving this system gives
\[
  b_1=-\frac{x+y}{2},\qquad
  b_2=\frac{x-y}{2},\qquad
  b_3=y.
\]
Substituting the angles into the expanded product, we get
\[
  V_3(t)
  =
  e^{i\frac{\pi}{2}H}e^{itb_1\psi_0}
  e^{-i\pi H}e^{itb_2\psi_0}
  e^{i\frac{\pi}{2}H}e^{itb_3\psi_0}.
\]
The final factor \(e^{-ia_3H}\) disappears because \(a_3=0\). This is exactly the \(6\)-factor retraction $\mathrm{R}_U^{(6)}(\eta)
  =
  e^{i\frac{\pi}{2}H}e^{ib_1\psi_0}
  e^{-i\pi H}e^{ib_2\psi_0}
  e^{i\frac{\pi}{2}H}e^{iy\psi_0}U$ introduced in the main text.

\subsection{Four building blocks: a 9-factor family}

For \(J=4\), we consider
\[
  V_4(t)
  =
  C_{a_1,b_1}(t)
  C_{a_2,b_2}(t)
  C_{a_3,b_3}(t)
  C_{a_4,b_4}(t).
\]
The matching equations can be written as
\[
  \begin{bmatrix}
  1 & 1 & 1 & 1\\
  -\sin a_1 & -\sin a_2 & -\sin a_3 & -\sin a_4\\
  \cos a_1 & \cos a_2 & \cos a_3 & \cos a_4
  \end{bmatrix}
  \begin{bmatrix}
  b_1\\ b_2\\ b_3\\ b_4
  \end{bmatrix}
  =
  \begin{bmatrix}
  0\\ x\\ y
  \end{bmatrix}.
\]
If the above \(3\times 4\) coefficient matrix has rank 3, then
the matching equations are solvable for every \((x,y)\), but the solutions \(b_1,b_2,b_3,b_4\) are not unique.
A general four block product expands to
\[
  V_4(t)
  =
  e^{ia_1H}e^{itb_1\psi_0}
  e^{i(a_2-a_1)H}e^{itb_2\psi_0}
  e^{i(a_3-a_2)H}e^{itb_3\psi_0}
  e^{i(a_4-a_3)H}e^{itb_4\psi_0}
  e^{-ia_4H}.
\]
Hence a generic four block construction gives a \(9\)-factor family.

The \(8\)-factor retraction introduced in the main text is obtained from a special fixed-angle choice. Specifically, choose
\[
  a_1=0,\qquad
  a_2=\frac{\pi}{2},\qquad
  a_3=-\frac{\pi}{2},\qquad
  a_4=-\pi.
\]
Then the matching matrix becomes
\[
  \begin{bmatrix}
  1 & 1 & 1 & 1\\
  0 & -1 & 1 & 0\\
  1 & 0 & 0 & -1
  \end{bmatrix}.
\]
For this fixed-angle choice, take
\[
  b_1=\frac y2,\qquad
  b_2=-\frac x2,\qquad
  b_3=\frac x2,\qquad
  b_4=-\frac y2.
\]
Then
\[
  \begin{bmatrix}
  1 & 1 & 1 & 1\\
  0 & -1 & 1 & 0\\
  1 & 0 & 0 & -1
  \end{bmatrix}
  \begin{bmatrix}
  y/2\\ -x/2\\ x/2\\ -y/2
  \end{bmatrix}
  =
  \begin{bmatrix}
  0\\ x\\ y
  \end{bmatrix}.
\]
Therefore the matching equations are satisfied, and hence $\dot V_4(0)=xX_0+yY_0.$ Finally, we have
\[
  V_4(t)
  =
  e^{it\frac y2\psi_0}
  e^{i\frac{\pi}{2}H}
  e^{-it\frac x2\psi_0}
  e^{-i\pi H}
  e^{it\frac x2\psi_0}
  e^{-i\frac{\pi}{2}H}
  e^{-it\frac y2\psi_0}
  e^{i\pi H}.
\]
The first \(H\)-generated factor $e^{ia_1H}$ is absent because \(a_1=0\).
Setting \(t=1\), we obtain the \(8\)-factor retraction introduced in the main text
\[
  \mathrm{R}_U^{(8)}(\eta)
  =
  e^{i\frac y2\psi_0}
  e^{i\frac{\pi}{2}H}
  e^{-i\frac x2\psi_0}
  e^{-i\pi H}
  e^{i\frac x2\psi_0}
  e^{-i\frac{\pi}{2}H}
  e^{-i\frac y2\psi_0}
  e^{i\pi H}U.
\]

\subsection{Summary}

The above constructions are all based on the same coefficient matching principle. Given \(\eta\in\mathcal WU\), write $\eta U^\dagger=xX_0+yY_0.$ We construct a product \(V(t;x,y)\) generated only by \(H\) and \(\psi_0\) such that $V(0;x,y)=I,$ $\dot V(0;x,y)=xX_0+yY_0.$ Then \(\mathrm{R}_U(\eta):=V(1;x,y)U\) is a retraction on the restricted tangent subspace \(\mathcal WU\).
The elementary block we used
\[
  C_{a,b}(t)=e^{iaH}e^{itb\psi_0}e^{-iaH}
\]
has initial velocity
\[
  \left.\frac{d}{dt}C_{a,b}(t)\right|_{t=0}
  =
  b\bigl(i\psi_0-\sin(a)X_0+(\cos(a)-1)Y_0\bigr).
\]
Thus, for a product of \(J\) such blocks, the retraction condition reduces to the coefficient matching equations
\[
  \begin{bmatrix}
  1 & 1 & \cdots & 1\\
  -\sin a_1 & -\sin a_2 & \cdots & -\sin a_J\\
  \cos a_1 & \cos a_2 & \cdots & \cos a_J
  \end{bmatrix}
  \begin{bmatrix}
  b_1\\ b_2\\ \vdots\\ b_J
  \end{bmatrix}
  =
  \begin{bmatrix}
  0\\ x\\ y
  \end{bmatrix}.
\]
The first row cancels the unwanted \(i\psi_0\) component, while the second and third rows match the prescribed \(X_0\) and \(Y_0\) components. One block cannot produce a nonzero tangent vector. Two blocks yield the \(5\)-factor family; three blocks yield the general \(7\)-factor family and, for a special choice of angles, the \(6\)-factor formula used in the main text. Four blocks give a general \(9\)-factor family, with the \(8\)-factor formula obtained from a particular fixed-angle choice.

\section{Cayley and Polar retractions preserve the Grover plane}

In this section, we show that, when the Cayley or polar retraction is used for \cref{pro-1}, the resulting update is equivalent to a Riemannian exponential update with a rescaled tangent direction. Hence, these retractions also preserve the Grover plane (i.e., all $\ket{\psi_k}$ lie in Grover plane). Nevertheless, their direct circuit implementations are not straightforward, which motivates the construction of a Grover-compatible retraction using the Grover-type gates.

We first recall the notation. We work on the unitary manifold,
\[
    \mathrm{U} (N)=\bigl\{ U\in \mathbb C^{N\times N} \big| U^{\dagger}U=I_N\bigr\}.
\]
Its tangent space at \(U\in \mathrm U(N)\) is $T_U=\mathfrak u(N)U,$ where
\[
\mathfrak u(N)=\{\Omega\in\mathbb C^{N\times N}:\Omega^\dagger=-\Omega\}
\]
is the space of skew-Hermitian matrices. Thus every tangent vector at \(U\) is written as \(\eta=\Omega U\), where \(\Omega\in\mathfrak u(N)\). The Riemannian exponential map is given by
\[
\operatorname{Exp}_U(\eta)=e^\Omega U.
\]
The Cayley retraction \cite{wen2013feasible} is
\begin{equation*}
\mathrm{R}_U^{\mathrm{cay}}(\eta)
=
\left(I-\frac12\Omega\right)^{-1}
\left(I+\frac12\Omega\right)U,
\end{equation*}
and polar retraction \cite{absil2012projection,absil2008optimization} is
\begin{equation*}
\mathrm{R}_U^{\mathrm{polar}}(\eta)
=(U+\eta)\left(I+\eta^{\dagger} \eta\right)^{-1 / 2}.
\end{equation*}

All functions of Hermitian matrices are defined by spectral calculus. In particular, if Hermitian \(A=Q\operatorname{diag}(\lambda_1,\ldots,\lambda_N)Q^\dagger\), then
\[
f(A)=Q\operatorname{diag}(f(\lambda_1),\ldots,f(\lambda_N))Q^\dagger.
\]
The quotients \(f(A)=2\arctan(A/2)/A\) and \(f(A)=\arctan A/A\) are understood by continuous extension at \(0\), with value \(1\).

\subsection{General skew-Hermitian directions}

We first record a general spectral form of the Cayley and polar retractions on the unitary group.  When written in terms of an arbitrary skew-Hermitian \(\Omega\), both retractions admit exponential representations via functional calculus.  In the special case where the nonzero spectrum of \(\Omega\) is \(\{\pm i\rho\}\), these representations reduce to the rescalings of the same tangent direction.  No Grover structure is used in this result.

\begin{theorem}
[General spectral form of Cayley and polar retractions]
\label{thm:general}
Let \(U\in\mathrm{U}(N)\), \(\Omega\in\mathfrak u(N)\), and \(\eta=\Omega U\). Set \(A :=-i\Omega\). Then \(A=A^\dagger\), and
\begin{equation*}
\mathrm{R}_U^{\mathrm{cay}}(\eta)
=
\exp\left(i\,2\arctan\frac A2\right)U,
\qquad
\mathrm{R}_U^{\mathrm{polar}}(\eta)
=
\exp(i\arctan A)U.
\end{equation*}
If all nonzero eigenvalues of \(A\) belong to \(\{\pm\rho\}\) for some \(\rho\ge0\),
or equivalently, all nonzero eigenvalues of \(\Omega\) belong to \(\{\pm i\rho\}\), then
\begin{align*}
\mathrm{R}_U^{\mathrm{cay}}(\eta)
&=
\exp\bigl(\gamma_C(\Omega)\Omega\bigr)U,
&
\gamma_C(\Omega)
&:=
\frac{\arctan(\rho/2)}{\rho/2},
\\
\mathrm{R}_U^{\mathrm{polar}}(\eta)
&=
\exp\bigl(\gamma_P(\Omega)\Omega\bigr)U,
&
\gamma_P(\Omega)
&:=
\frac{\arctan\rho}{\rho}.
\end{align*}
At \(\rho=0\), both factors are defined by continuous extension and equal \(1\).
\end{theorem}

\begin{proof}
We write the spectral decomposition of Hermitian $A$ as
\[
A=Q\operatorname{diag}(\lambda_1,\ldots,\lambda_N)Q^\dagger,
\qquad
\lambda_j\in\mathbb R,
\qquad
Q\in \mathrm U(N).
\]
Since \(A=-i\Omega\), we have \(\Omega=iA\). For the Cayley retraction, we have
\begin{align*}
\left(I-\frac12\Omega\right)^{-1}
\left(I+\frac12\Omega\right)
&=
\left(I-\frac12 iA \right)^{-1}
\left(I+\frac12 iA\right) \\
&=
Q\operatorname{diag}\left(
\frac{1+i\lambda_j/2}{1-i\lambda_j/2}
\right)Q^\dagger \\
&=
Q\operatorname{diag}\left(
e^{2i\arctan(\lambda_j/2)}
\right)Q^\dagger \\
&=
\exp\left(i\,2\arctan\tfrac A2\right),
\end{align*}
where we used $\frac{1+it}{1-it}=e^{2i\arctan t}$ for any $t\in\mathbb R.$\footnote{Indeed, let \(\theta :=\arctan t\). Then \(\theta\in(-\pi/2,\pi/2)\), \(t=\tan\theta\). By Euler's formula, \(\cos\theta\pm i\sin\theta=e^{\pm i\theta}\).
Hence,
\(\frac{1+it}{1-it}
=\frac{1+i\tan\theta}{1-i\tan\theta}
=\frac{\cos\theta+i\sin\theta}{\cos\theta-i\sin\theta}=e^{2i\theta}=e^{2i\arctan t}\).} Hence,
\[
\mathrm{R}_U^{\mathrm{cay}}(\eta)
=
\exp\left(i\,2\arctan\frac A2\right)U.
\]
For the polar retraction, since \(U+\eta=(I+\Omega)U\) and $I+\eta^\dagger\eta = U^\dagger(I+\Omega^\dagger\Omega)U,$ we get
\[
(I+\eta^\dagger\eta)^{-1/2}
=(U^\dagger(I+\Omega^\dagger\Omega)U)^{-1/2}
=
U^\dagger(I+\Omega^\dagger\Omega)^{-1/2}U.
\]
Therefore
\begin{align*}
\mathrm{R}_U^{\mathrm{polar}}(\eta)
&=
(I+\Omega)(I+\Omega^\dagger\Omega)^{-1/2}U  \\
&=
(I+iA)(I+A^2)^{-1/2}U  \\
&=
Q\operatorname{diag}\left(
\frac{1+i\lambda_j}{\sqrt{1+\lambda_j^2}}
\right)Q^\dagger U  \\
&=
Q\operatorname{diag}\left(
e^{i\arctan\lambda_j}
\right)Q^\dagger U  \\
&=
\exp(i\arctan A)U,
\end{align*}
where we used $\frac{1+it}{\sqrt{1+t^2}}=e^{i\arctan t}$ for any $t\in\mathbb R.$\footnote{Indeed, let \(\theta :=\arctan t\). Then \(\theta\in(-\pi/2,\pi/2)\), \(t=\tan\theta\), and \(\cos\theta=1/\sqrt{1+t^2}\) and \(\sin\theta=t/\sqrt{1+t^2}\). Therefore, \(\frac{1+it}{\sqrt{1+t^2}}=\cos\theta+i\sin\theta=e^{i\theta}=e^{i\arctan t}\).}

Now assume that all nonzero eigenvalues $\lambda_j$ of \(A\) belong to
\(\{\rho,-\rho\}\). If \(\rho>0\), then every eigenvalue $\lambda_j$ of \(A\) belongs to
\(\{0,\rho,-\rho\}\). Hence, by the oddness of \(\arctan (\cdot)\),
\[
2\arctan\frac{\lambda_j}{2}
=
\gamma_C(\Omega)\lambda_j,
\qquad
\gamma_C(\Omega)
:=
\frac{\arctan(\rho/2)}{\rho/2},
\]
for every \(j\). Thus $2\arctan\frac A2=\gamma_C(\Omega)A.$ Multiplying by \(i\) and using \(\Omega=iA\), we obtain
\[
i\,2\arctan\frac A2=\gamma_C(\Omega)\Omega.
\]
Therefore, $\mathrm{R}_U^{\mathrm{cay}}(\eta) = \exp\bigl(\gamma_C(\Omega)\Omega\bigr)U.$ Similarly, $\mathrm{R}_U^{\mathrm{polar}}(\eta) = \exp\bigl(\gamma_P(\Omega)\Omega\bigr)U.$
If \(\rho=0\), then \(A=0\), hence \(\Omega=0\), and both retractions equal \(U\). The two factors are defined by continuous extension, with
\[
\lim_{\rho\to0}\frac{\arctan(\rho/2)}{\rho/2}
=
\lim_{\rho\to0}\frac{\arctan\rho}{\rho}
=
1.
\]
Thus the formulas also hold for \(\rho=0\).
\end{proof}

\subsection{The Grover two-dimensional gradient subspace}

We now specialize to the Grover setting of interest.
Let \(H\in\mathbb C^{N\times N}\) be an orthogonal projector, so \(H=H^\dagger=H^2\).
Let \(|\psi_0\rangle\in\mathbb C^N\) be normalized, set \(\psi_0=|\psi_0\rangle\langle\psi_0|\), and assume the nontrivial regime \(0<\langle\psi_0|H|\psi_0\rangle<1\). Define
\[
X_0:=[H,\psi_0],
\qquad
Y_0:=i[H,X_0],
\qquad
\mathcal W:=\operatorname{span}_{\mathbb R}\{X_0,Y_0\}\subset\mathfrak u(N).
\]

The following theorem shows that, when the tangent direction is restricted to the Grover two-dimensional gradient subspace \(\mathcal W\), both the Cayley and polar retractions reduce to the Riemannian exponential map with a scalar rescaling of the tangent vector. The key point is that every \(\Omega\in\mathcal W\) acts nontrivially only on the Grover plane $\mathcal{S}=\operatorname{span}_{\mathbb{C}}\left\{\left|\psi_0\right\rangle, H\left|\psi_0\right\rangle\right\}$, where its spectrum has $\{ \pm i \rho\}$.

\begin{theorem}[Cayley and polar retractions on the Grover gradient subspace]
\label{thm:grover}
Let \(U\in\mathrm{U}(N)\), \(\Omega\in\mathcal W\), and \(\eta=\Omega U\). Then
\begin{align}
\mathrm{R}_U^{\mathrm{cay}}(\eta)
&=
\operatorname{Exp}_U(\gamma_C(\Omega)\eta),
&
\gamma_C(\Omega)
&=
\frac{
\arctan\left(\|\Omega\|_F/(2\sqrt2)\right)
}{
\|\Omega\|_F/(2\sqrt2)
},
\label{eq:cayley-grover}
\\
\mathrm{R}_U^{\mathrm{polar}}(\eta)
&=
\operatorname{Exp}_U(\gamma_P(\Omega)\eta),
&
\gamma_P(\Omega)
&=
\frac{
\arctan\left(\|\Omega\|_F/\sqrt2\right)
}{
\|\Omega\|_F/\sqrt2
}.
\label{eq:polar-grover}
\end{align}
At \(\Omega=0\), both scaling factors are defined by continuous extension and equal \(1\).
\end{theorem}

\begin{proof}
Let $u:=H|\psi_0\rangle,$ $v:=(I-H)|\psi_0\rangle.$ Then \(|\psi_0\rangle=u+v\), \(Hu=u\), \(Hv=0\), and \(u\perp v\). Set \(q_0 :=\|u\|^2\) and \(p_0:=\|v\|^2=1-q_0\). The Grover plane is $\mathcal S=\operatorname{span}_{\mathbb C}\{u,v\}.$ We now compute the action of \(\Omega\) on the Grover plane \(\mathcal S\). In the non-normalized orthogonal basis $\mathcal{B}:=\{u, v\}$ of \(\mathcal S\), the matrices of \(H\) and \(\psi_0\) restricted to \(\mathcal S\) are
\[
[H]_{\mathcal{B}}=\begin{pmatrix}
1 & 0 \\
0 & 0
\end{pmatrix},
\qquad
\left[\psi_0\right]_{\mathcal{B}}=
\begin{pmatrix}
q_0&p_0\\
q_0&p_0
\end{pmatrix}.
\]
Therefore
\begin{align*}
\left[X_0\right]_{\mathcal{B}}=\left[H, \psi_0\right]_{\mathcal{B}}=
\begin{pmatrix}
0&p_0\\
-q_0&0
\end{pmatrix},
\qquad
\left[Y_0\right]_{\mathcal{B}}=i\left[H, X_0\right]_{\mathcal{B}}
=
\begin{pmatrix}
0&ip_0\\
iq_0&0
\end{pmatrix}.
\end{align*}
For \(\Omega=xX_0+yY_0\), define \(\alpha:=x+iy\). Then
\[
[\Omega]_{\mathcal{B}}=x [X_0]_{\mathcal{B}} +y [Y_0]_{\mathcal{B}}=
\begin{pmatrix}
0&p_0\alpha\\
-q_0\overline{\alpha}&0
\end{pmatrix}.
\]
A direct multiplication gives
\begin{align*}
[\Omega^2]_{\mathcal{B}}=
\begin{pmatrix}
0&p_0\alpha\\
-q_0\overline{\alpha}&0
\end{pmatrix}^2
=
-p_0q_0|\alpha|^2
\begin{pmatrix}
1&0\\
0&1
\end{pmatrix}.
\end{align*}
Thus, on \(\mathcal S\), the operator \(\Omega\) satisfies
\[
\left.\Omega^2\right|_{\mathcal S}=-p_0q_0|\alpha|^2 I_{\mathcal S}.
\]

It remains to verify the spectral situation in Theorem~\ref{thm:general}. By \cref{lem-subspace-inv}, both \(X_0\) and \(Y_0\) vanish on \(\mathcal S^\perp\). Hence \(\Omega\in\mathcal W\) also vanishes on \(\mathcal S^\perp\). Consequently, the only possible nonzero eigenvalues of \(\Omega\) come from its restriction to \(\mathcal S\). From the identity above, these eigenvalues are
\[
\pm i\sqrt{p_0q_0}\,|\alpha|.
\]
Moreover, $\|\eta\|_F^2=\|\Omega\|_F^2=2p_0q_0|\alpha|^2.$ Hence, let
\[
\rho :=\sqrt{p_0q_0}|\alpha|=\frac{\|\Omega\|_F}{\sqrt2}.
\]
Applying Theorem~\ref{thm:general} with this \(\rho\) gives exactly \eqref{eq:cayley-grover} and \eqref{eq:polar-grover}.
\end{proof}

\section{QR retraction does not preserve the Grover plane}

Let \(U \in \mathrm{U}(N)\), \(\Omega \in \mathfrak{u}(N)\), and \(\eta=\Omega U\), where $\mathfrak{u}(N) = \left\{\Omega \in \mathbb{C}^{N \times N}: \Omega^{\dagger}=-\Omega\right\}.$
The QR retraction \cite{sato2019cholesky,absil2008optimization} is defined by
\[
\mathrm{R}_U^{\mathrm{qr}}(\eta)=\operatorname{qf}(U+\eta),
\]
where \(\operatorname{qf}\) denotes the \(Q\)-factor in the QR decomposition, with the diagonal entries of the \(R\)-factor chosen to be positive. Although this is a standard retraction on the unitary manifold, it does not, in general, preserve the Grover plane (i.e., all $\ket{\psi_k}$ lie in Grover plane). We demonstrate this by an explicit counterexample.

Let
\[
N=3,\qquad
H=
\begin{pmatrix}
1&0&0\\
0&0&0\\
0&0&0
\end{pmatrix},
\qquad
\ket{\psi_0}
=
\frac{1}{\sqrt3}
\begin{pmatrix}
1\\1\\1
\end{pmatrix}.
\]
The choice \(N=3\) is only for simplicity; analogous examples can be embedded into higher dimensional spaces, including dimensions \(N=2^n\). Define
\[
u:=H\ket{\psi_0}
=
\frac{1}{\sqrt3}
\begin{pmatrix}
1\\0\\0
\end{pmatrix},
\qquad
v:=(I-H)\ket{\psi_0}
=
\frac{1}{\sqrt3}
\begin{pmatrix}
0\\1\\1
\end{pmatrix}.
\]
Thus the Grover plane is
\[
\mathcal S
:=
\operatorname{span}_{\mathbb C}
\left\{
u,v
\right\}
=
\operatorname{span}_{\mathbb C}
\left\{
\begin{pmatrix}
1\\0\\0
\end{pmatrix},
\begin{pmatrix}
0\\1\\1
\end{pmatrix}
\right\},
\]
and its orthogonal complement is
\[
\mathcal S^\perp
=
\operatorname{span}_{\mathbb C}\{n\},
\qquad
n:=
\frac{1}{\sqrt2}
\begin{pmatrix}
0\\1\\-1
\end{pmatrix}.
\]
For this choice of \(H\) and \(\psi_0=\ket{\psi_0}\bra{\psi_0}\), we have
\[
X_0=[H,\psi_0]
=
\frac13
\begin{pmatrix}
0&1&1\\
-1&0&0\\
-1&0&0
\end{pmatrix}.
\]
Take \(U=I\) and \(\Omega=tX_0\in\mathcal W\), with \(t\neq0\). Set \(a:=t/3\). Then
\[
I+\Omega
=
\begin{pmatrix}
1&a&a\\
-a&1&0\\
-a&0&1
\end{pmatrix}.
\]
Let \(Q:=\operatorname{qf}(I+\Omega)\), where \(\operatorname{qf}\) denotes the \(Q\)-factor in the standard QR decomposition, with the diagonal entries of \(R\)-factor chosen to be positive. The columns of \(I+\Omega\) are
\[
c_1=
\begin{pmatrix}
1\\-a\\-a
\end{pmatrix},
\qquad
c_2=
\begin{pmatrix}
a\\1\\0
\end{pmatrix},
\qquad
c_3=
\begin{pmatrix}
a\\0\\1
\end{pmatrix}.
\]
Applying the Gram--Schmidt procedure gives
\[
q_1=
\frac{1}{\sqrt{1+2a^2}}
\begin{pmatrix}
1\\-a\\-a
\end{pmatrix},
\qquad
q_2=
\frac{1}{\sqrt{1+a^2}}
\begin{pmatrix}
a\\1\\0
\end{pmatrix},
\qquad
q_3
=
\sqrt{\frac{1+a^2}{1+2a^2}}
\begin{pmatrix}
\dfrac{a}{1+a^2}\\[3pt]
-\dfrac{a^2}{1+a^2}\\[3pt]
1
\end{pmatrix}.
\]
Hence \(Q=(q_1,q_2,q_3)\), and
\[
Q\ket{\psi_0}
=
\frac{1}{\sqrt3}(q_1+q_2+q_3).
\]
We now compute the component of \(Q\ket{\psi_0}\) along \(\mathcal S^\perp\). First,
\[
\langle n,q_1\rangle=0,
\qquad
\langle n,q_2\rangle=
\frac{1}{\sqrt{2(1+a^2)}},
\qquad
\langle n,q_3\rangle
=
-\frac{\sqrt{1+2a^2}}{\sqrt{2(1+a^2)}}.
\]
Therefore
\begin{align*}
\left\langle n,Q | \psi_0 \right\rangle
=
\frac{1}{\sqrt3}
\left(
\langle n,q_1\rangle+\langle n,q_2\rangle+\langle n,q_3\rangle
\right)
=
\frac{1-\sqrt{1+2a^2}}{\sqrt{6(1+a^2)}}.
\end{align*}
Since \(a\neq0\), we have \(\sqrt{1+2a^2}>1\), and hence $\left\langle n,Q | \psi_0 \right\rangle \neq0.$ Thus
\[
Q\ket{\psi_0}\notin\mathcal S.
\]
Equivalently, the QR update
\[
\mathrm{R}_I^{\mathrm{qr}}(\Omega)\ket{\psi_0}
=
\operatorname{qf}(I+\Omega)\ket{\psi_0}
=
Q\ket{\psi_0}
\notin \mathcal S.
\]
Thus the QR retraction update leaves the Grover plane, even though \(\Omega\in\mathcal W\). This shows that the QR retraction does not, in general, preserve the Grover 2D dynamics.

\bibliographystyle{siamplain}
\bibliography{references}